\documentclass[journal]{IEEEtran}

\usepackage{cite}          
\usepackage{amsmath,amssymb,amsfonts,amsthm}
\newtheorem{proposition}{Proposition}
\newtheorem{remark}{Remark}
\usepackage{algorithm}
\usepackage{mathtools}
\usepackage{algorithmic}
\usepackage{subcaption}
\usepackage{booktabs}
\usepackage{multirow}
\usepackage{xcolor}
\usepackage{colortbl}
\usepackage{graphicx}
\usepackage{caption}
\usepackage{enumitem}
\usepackage{url}
\usepackage{hyperref}
\usepackage{cleveref}
\newcommand{\revise}[1]{\textcolor{black}{#1}}

\begin{document}

\title{MIDA: Multi-Interest Decoding with Adaptive Aggregation for Session-based Recommendation}

\author{Linjiang~Guo,
        Nitin~Bisht,
        Shiqing~Wu,
        and~Guandong~Xu%
\thanks{Linjiang~Guo and Nitin~Bisht contribute equally to this work.}%
\thanks{L. Guo and N. Bisht are with the University of Technology Sydney, Sydney, Australia.
E-mail: \{linjiang.guo, nitin.bisht\}@student.uts.edu.au.}%
\thanks{S. Wu is with the City University of Macau, Macau SAR. E-mail: sqwu@cityu.edu.mo.}%
\thanks{G. Xu is with the Education University of Hong Kong, Hong Kong SAR. E-mail: gdxu@eduhk.hk.}%
\thanks{S. Wu and G. Xu are corresponding authors.}}

\maketitle

{\footnotesize
\noindent
\textcopyright{} 2026 IEEE.  Personal use of this material is permitted.  Permission from IEEE must be obtained for all other uses, in any current or future media, including reprinting/republishing this material for advertising or promotional purposes, creating new collective works, for resale or redistribution to servers or lists, or reuse of any copyrighted component of this work in other works.

\par\medskip
\noindent
This is the author-accepted manuscript of an article
published in \emph{IEEE Transactions on Knowledge
and Data Engineering (TKDE)}.

The final published version is available at
\href{https://ieeexplore.ieee.org/document/11738675/}
{IEEE Xplore}.

DOI: \href{https://doi.org/10.1109/TKDE.2026.3740693}{10.1109/TKDE.2026.3740693}.
\par\medskip
}

\begin{abstract}
Session-based recommendation (SBR) infers user interest from limited
interaction sequences, yet current models show a consistent performance
drop when longer histories are available.  This degradation is associated
with decoder designs that overemphasize recent items while underutilizing
distant, yet semantically relevant, interactions, leading to poor modeling of
long-range dependencies and multi-interest behaviors.  To address this
challenge, we introduce \textbf{MIDA} (\textbf{M}ulti-\textbf{I}nterest 
\textbf{D}ecoding with \textbf{A}daptive Aggregation for Session-based Recommendation), a
lightweight and encoder-agnostic framework that provides a reformulation of
the decoding stage.  The central design principle is a \emph{two-space decoupling}: clustering
operates on raw item embeddings so that semantic groupings are not distorted
by encoder-learned relational biases; encoder outputs are reserved for
aggregation, where sequential context remains intact.  Concretely, MIDA
(i)~employs differentiable multi-interest clustering on raw embeddings,
preserving semantic structure that encoder representations would otherwise
obscure, and (ii)~fuses stability and plasticity branches through a gate
conditioned on average cluster size, so that the inferred clustering
structure directly governs aggregation behavior through a structural coupling
across the two stages.
Extensive experiments on multiple benchmark datasets demonstrate that MIDA
consistently outperforms state-of-the-art methods, with meaningful accuracy
improvements while remaining computationally practical.
\end{abstract}

\begin{IEEEkeywords}
Session-Based Recommendations, Dynamic Aggregation, Multi-Interest,
Session Representation
\end{IEEEkeywords}

\section{Introduction}
\label{sec:intro}

Recommendation systems (RS)~\cite{covington2016deep,wang2020time,wu2022graph}
have become ubiquitous in modern online platforms, delivering personalized
and timely content by learning from user behavior.  In practice, persistent
user profiles are often unavailable, or only interactions within limited time
windows can be observed.  Session-based Recommendation
(SBR)~\cite{guo2019streaming,li2024graph,choi2025linear} provides an
effective solution: it infers user interest from the sequence of interactions
in a single session and predicts the next item of interest, with particular
relevance to e-commerce, news, and video streaming, where users browse
anonymously or interact through temporary identifiers.

\begin{figure}[t]
  \centering
  \begin{subfigure}{0.49\columnwidth}
    \includegraphics[width=\columnwidth]{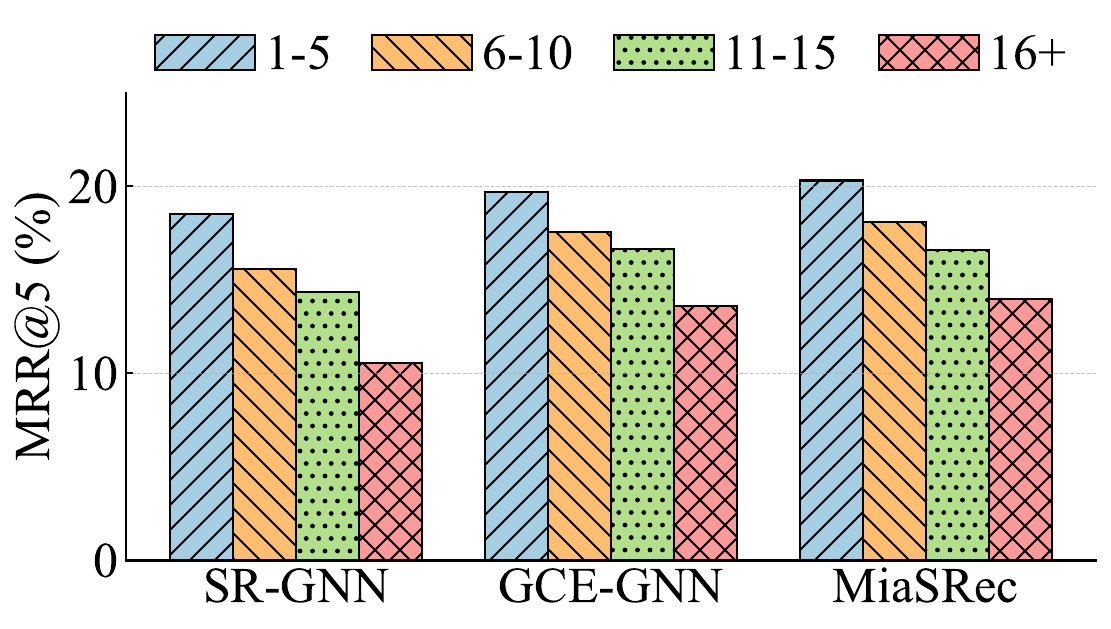}
  \end{subfigure}
  \hfill
  \begin{subfigure}{0.49\columnwidth}
    \includegraphics[width=\columnwidth]{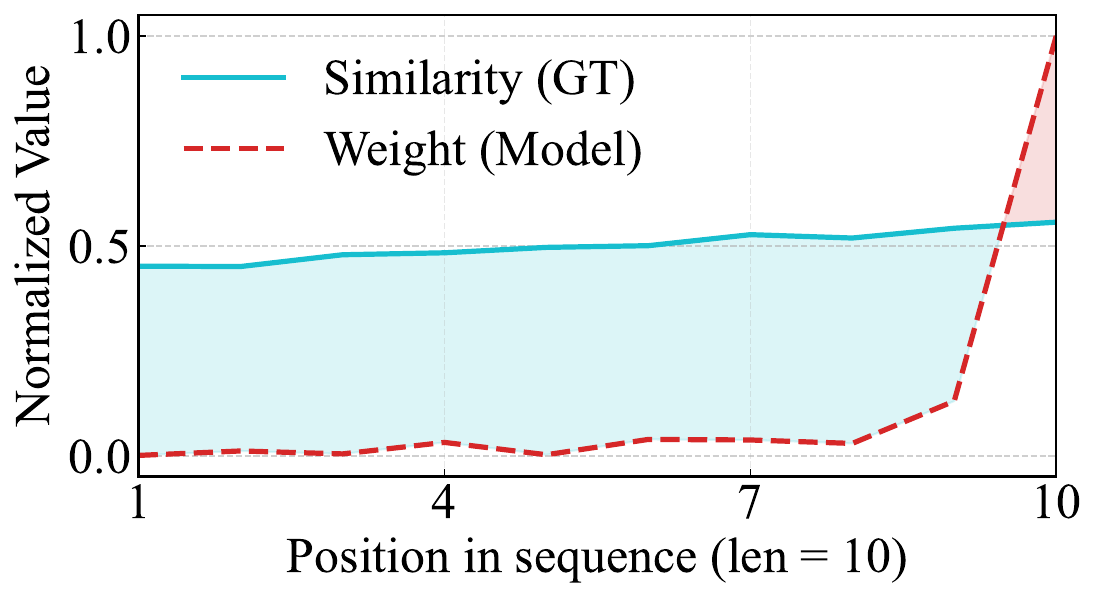}
  \end{subfigure}
  \caption{Analysis of decoder limitations in existing session-based RS.
    \emph{Top}: MRR@5 as a function of session length across state-of-the-art
    models, with longer sessions yielding lower accuracy despite containing richer
    information.
    \emph{Bottom}: Attention weight distribution (red dashed) versus
    ground-truth item cosine similarity to the target (blue solid) in a
    length-10 session, providing evidence for a strong recency bias.
    Results are generated on the Diginetica dataset using an SR-GNN backbone;
    attention weights correspond to the single-layer decoder attention.}
  \label{fig:empirical}
\end{figure}

Recently, SBR has been advanced through graph-based~\cite{wu2019session,zhang2020personalized}
and transformer-based architectures~\cite{zhuo2024multi} that adopt an
encoder-decoder framework, enabling more effective modeling of complex item
dependencies.  Prior research concentrates heavily on the encoder mechanism
to capture both local and global relationships among
items~\cite{zheng2024heterogeneous,wang2024spatial,ou2025ls}. \revise{Related multi-intent methods further enhance session representations by constructing or selecting multiple intent representations within model-specific architectures~\cite{choi2024multi,li2022enhancing,shen2023temporal,lv2025dynamic}.} The decoder
mechanism, which ultimately determines how encoded information is transformed
into user preference representations, has received comparatively little
attention.

Some recent research utilizes pooling-based
methods~\cite{arandjelovic2016netvlad,choi2024multi}, which compress all
items into a fixed-length embedding.  While simple and efficient, this
approach inevitably discards sequential and item-level distinctions, becoming
particularly harmful in long sessions where diverse user interests emerge.  As
a result, pooling-based decoders struggle to preserve the temporal information
necessary for accurate next-item prediction.  Another line of work focuses on
point-wise approaches~\cite{ou2025ls}, which rely on attention or gating
anchored on the last item to balance short-term and long-term intents.
However, this design overweights recent interactions while underutilising
earlier but semantically relevant ones, ultimately failing to capture the
multi-interest nature of long sessions.

These limitations raise an important concern: do existing decoders truly
exploit the richer historical context available in long sessions?  To
investigate, we conduct an empirical analysis on representative models, as
shown in Fig.~\ref{fig:empirical}.  Point-wise approaches assign
disproportionate weight to the last few items while ignoring earlier ones;
pooling-based approaches blur temporal signals across the sequence.  Both
weaknesses manifest in a consistent empirical pattern: when sessions are
grouped by length, longer sessions produce lower accuracy across
state-of-the-art models, despite containing richer historical information.
Notably, this degradation is observed across architecturally distinct model
families (GRU-based, graph-based, and Transformer-based) that differ in
every encoder design dimension yet share the same decoder paradigm.  This
cross-architecture consistency cannot be attributed to any specific encoder
design; rather, it implicates the shared decoder mechanism as the primary
source of underperformance on longer sessions.  To isolate the decoder's
contribution, we adopt a controlled evaluation protocol (Section~\ref{sec:experiments}):
we hold each encoder fixed and vary only the decoder component across all
baselines, enabling a direct assessment of decoder-side impact while factoring
out encoder differences.  The central research question is therefore:
\emph{How can we design a lightweight, encoder-agnostic decoder that
effectively exploits diverse historical contexts to model user interests?}

To answer this, we introduce \textbf{MIDA} (Multi-Interest Decoding with Adaptive
Aggregation for Session-based Recommendation) \footnote{Our code is available at: \url{https://github.com/downw/MIDA_Rec}}, a lightweight and
encoder-agnostic framework that provides a reformulation of the decoding
stage to tackle the above limitations by modeling both persistent and
evolving aspects of user intent.  The two empirical failure modes of existing decoders (pooling methods that
conflate semantically distinct items and point-wise methods that discard all
but the most recent interactions) translate into two concrete design challenges: (a)~How can we
capture diverse persistent interests that are semantically related but
temporally distant?  (b)~How can we balance stability (persistent intent)
with plasticity (evolving intent) in user behavior?

To address these challenges, MIDA employs \emph{dynamic multi-interest
clustering}, which groups semantically related items via differentiable
similarity thresholds, thereby preserving long-range dependencies and
capturing diverse user preferences that pooling-based methods collapse.  It
also introduces \emph{adaptive aggregation}, which disentangles user intent
into stability and plasticity and fuses them through a gating mechanism
conditioned on both branch representations and semantic cohesion signals
(i.e.\ average cluster size), enabling fine-grained adaptation to consistent
and evolving behavioral signals.  Together, these components yield a decoder
that efficiently utilizes history information, retains temporal structure, and
models the multi-interest nature of complex sessions.  The framework is
illustrated in Fig.~\ref{fig:framework}.

The contributions of this work are as follows.
\begin{enumerate}[leftmargin=*]
  \item We identify the \emph{decoder bottleneck} in SBR, namely the systematic
        tendency of existing decoders to underutilize long-range historical
        context, and provide a lightweight, encoder-agnostic reformulation of
        the decoding stage that directly addresses this limitation.
  \item A differentiable clustering mechanism with regularisation is proposed
        for dynamic interest extraction, together with a stability--plasticity
        aggregation scheme that disentangles persistent and evolving user
        intents, emphasizing the \emph{decoder-side perspective}.
  \item The framework incorporates a two-stage decoder design in which semantic
        similarity computation is decoupled from relational representation:
        Dynamic Multi-Interest Clustering operates on raw embeddings to
        identify similar items across temporal distances, while Adaptive
        Aggregation uses encoder outputs to construct attention-based session
        representations that capture sequential dependencies.
        The decoder is encoder-agnostic: it replaces only the final
        aggregation step of any backbone without modifying encoder parameters.
  \item Comprehensive experiments on multiple benchmark datasets show that
        MIDA outperforms strong baselines across diverse encoder architectures.
\end{enumerate}

\section{Related Work}
\label{sec:related}

\paragraph{Session-based recommendation.}
Traditional SBR methods typically rely on heuristic or shallow models such as
Markov chains to model item
transitions~\cite{rendle2010factorizing}.  Recent studies increasingly focus
on deep neural architectures, including recurrent
networks~\cite{quadrana2017personalizing,ren2019repeatnet}, attention
mechanisms~\cite{li2017neural}, and graph neural
networks~\cite{wu2019session,xia2021self,choi2025linear}, which effectively
encode complex transition dependencies within sessions.  Several works aim to
integrate both local and global item relations~\cite{qiu2019rethinking}, or
leverage contrastive learning to enhance representation
robustness~\cite{xia2021self}.  Further explorations extend to
Transformer-based frameworks~\cite{wilm2023scaling,zheng2024harnessing} and
memory-augmented
architectures~\cite{wang2023exploiting,li2023repetition}.  Side information
has also emerged as a promising direction, including temporal
signals~\cite{wan2023spatio}, item
descriptions~\cite{zheng2024harnessing}, user
behaviors~\cite{yang2022multi,han2024efficient}, and category
information~\cite{wang2024knowledge}.  Across these studies, the encoder
receives the bulk of attention, while the decoder remains largely fixed.

\paragraph{Multi-interest modeling.}
There is growing interest in capturing multi-interest structures within
sessions.  One approach segments item embeddings into multiple parts, each
representing a distinct interest~\cite{li2022enhancing}.  Another treats
items as multi-aspect entities mapped into multiple candidate spaces
corresponding to different aspects~\cite{choi2024multi,shen2023temporal}.
Dynamic interest adaptation mechanisms for automatically inferring the number
of interests have also been
explored~\cite{lv2025dynamic}.  MIDA is conceptually related to these lines
of work, but differs in two important ways.  First, it employs a dynamic,
differentiable clustering mechanism based on within-session semantic
similarity, grouping interests adaptively without a preset number or heavy
overhead.  Second, it applies multi-interest modeling at the \emph{decoder
stage} and further disentangles persistent from evolving intents, enabling
richer utilisation of historical information.

\paragraph{Decoder methods in recommendation.}
Decoder-based recommendation methods differ in how user preferences are
generated.  Session-based decoders may form a user interest representation
from the last interacted item~\cite{pan2020star}, global session
aggregation~\cite{zheng2020dgtn}, or temporal exponential decay
sum~\cite{deng2022g}.  CTR prediction decoders score each item in a fixed
candidate set via binary classifiers~\cite{si2024twin}.  Sequential
recommendation approaches such as LightGCN~\cite{he2020lightgcn} bypass
explicit decoding by computing relevance via neighborhood embedding
aggregation.  Recent work employs Deep Q-Network agents to optimize graph
connectivity~\cite{zeng2024fair} or multi-scale rescaling of positional
indices to mitigate the ``lost-in-the-middle'' problem~\cite{zhang2024found}.
Accurately capturing user interests is especially important because scoring
every candidate item is impractical; imprecise interest representations can
mislead the decoder~\cite{klenitskiy2024does}, particularly when user intents
shift during a session~\cite{wang2024comprehensive,yoo2025embracing}.  MIDA
addresses these challenges by reformulating the decoder as a dual-branch
system grounded in within-session semantic structure.

Unlike existing methods that either treat the encoder and decoder jointly or
focus on encoder redesign, MIDA treats the decoder as a {first-class
component} while remaining encoder-agnostic.  Specifically, MIDA replaces
only the final aggregation module of a backbone (e.g.\ the attention-based
pooling in SR-GNN/GCE-GNN or the output projection in GRU4Rec), preserving
all encoder parameters and training schedules.  This modular design enables
straightforward integration across heterogeneous encoder families.  The combination of differentiable threshold-based clustering (without
a pre-specified number of interests), a stability and plasticity dual-query design,
and a semantically conditioned gating mechanism distinguishes MIDA from
multi-interest methods that operate at the encoder level and from decoder
methods that rely on a single fixed aggregation rule.  In particular, unlike
AttenMixer~\cite{zhang2023efficiently}, which introduces multi-level propagation
layers that can conflict with rich encoders, MIDA's decoupled query construction
explicitly separates semantic similarity (raw embeddings) from relational context
(encoder outputs), avoiding gradient interference while retaining the benefits of
both representations.  The essential novelty of MIDA lies not in assembling
independently existing modules but in a principled {two-space decoupling}:
clustering operates on raw item embeddings so that semantic groupings are not
distorted by the relational biases the encoder has learned; encoder outputs
are reserved for aggregation, preserving the sequential dependencies the
encoder provides.
This is an architectural consequence of separating {which} items are related from {how} they should be aggregated given session
context, a distinction no existing decoder draws explicitly.

\section{The Proposed Framework}
\label{sec:method}

We propose \textbf{MIDA} (\textbf{M}ulti-\textbf{I}nterest and
A\textbf{d}aptive \textbf{A}ggregation for Session-based Recommendation), a
lightweight decoder framework that models both long-range, multi-interest
dependencies and short-term preference shifts in SBR.  An overview is
illustrated in Fig.~\ref{fig:framework}.

\begin{figure*}[htbp]
  \centering
  \includegraphics[width=0.9\textwidth]{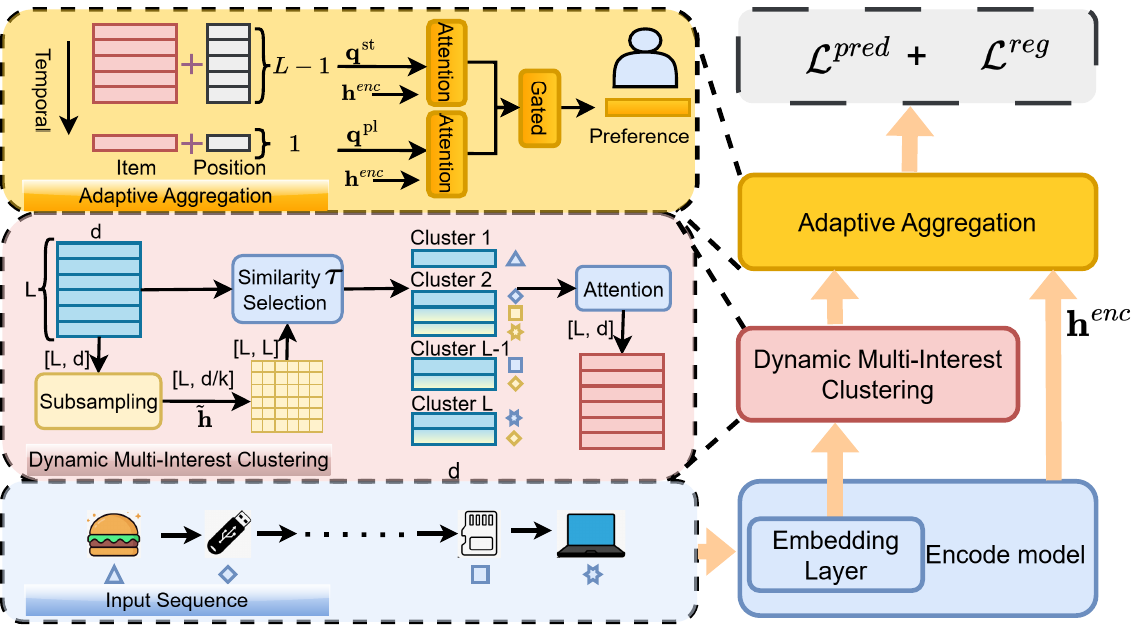}
  \caption{Overview of the proposed MIDA framework. The model first extracts
    raw embeddings from the embedding layer, which are clustered by the
    Dynamic Multi-Interest Clustering module to retain semantically related
    yet temporally distant signals under a learned threshold~$\tau$.  The
    Adaptive Aggregation module then decouples these signals into plasticity
    (evolving) and stability (persistent) intents, which interact with encoder
    outputs and are adaptively fused via a gating mechanism for final
    prediction.}
  \label{fig:framework}
\end{figure*}

\subsection{Notation}
\label{subsec:notation}

Table~\ref{tab:notation} summarises the key symbols used throughout this
paper.

\begin{table}[t]
  \caption{Key Notation}
  \label{tab:notation}
  \centering
  \begin{tabular}{cl}
    \toprule
    \textbf{Symbol} & \textbf{Meaning} \\
    \midrule
    $S$                       & Session (ordered item sequence) \\
    $L$                       & Session length ($|S|$) \\
    $d$                       & Embedding dimension \\
    $\tau$                    & Similarity threshold (global learnable parameter) \\
    $k$                       & {Embedding-dimension subsampling stride} \\
    $\mathbf{h}_i \in \mathbb{R}^d$           & Raw item embedding \\
    $\mathbf{h}^{\mathrm{enc}}_i \in \mathbb{R}^d$ & Encoder output embedding \\
    $\mathbf{M} \in \mathbb{R}^{L\times L}$  & Pairwise cosine similarity matrix \\
    $w_{ij} \in \mathbb{R}$   & Soft cluster weight between items $i$ and $j$ \\
    $C_i$                     & Cluster of item $v_i$ \\
    $e$                        & Average cluster size ($\bar{c} = \frac{1}{L}\sum_i|C_i|$) \\
    \revise{$\tilde{\mathbf{h}}'_i \in \mathbb{R}^{\lceil d/k\rceil}$} & {Subsampled, $\ell_2$-normalized embedding} \\
    $\hat{\mathbf{h}}_i \in \mathbb{R}^d$ & Cluster-enhanced embedding of item $v_i$ \\
    $\mathbf{q}^{st} \in \mathbb{R}^d$ & Stability query (sums over items $1,\dots,L{-}1$) \\
    $\mathbf{q}^{pl} \in \mathbb{R}^d$ & Plasticity query (anchored to last item $v_L$) \\
    $\mathbf{h}_S \in \mathbb{R}^d$    & Session-level representation \\
    \bottomrule
  \end{tabular}
\end{table}

\subsection{Problem Formulation}
\label{subsec:problem}

Given a session $S = [v_1, v_2, \dots, v_L]$, where each $v_i \in \mathcal{V}$
is the $i$-th interacted item and $L$ is the session length, each item $v_i$
is mapped to a raw embedding $\mathbf{h}_i \in \mathbb{R}^d$ that captures
semantic information, where $d$ is the embedding dimension.  The encoder then
produces a contextualised embedding $\mathbf{h}^{\mathrm{enc}}_i \in
\mathbb{R}^d$ that captures sequential dependencies.  The task of SBR is to
predict the next item $v_{L+1}$ given the history $S$:
\begin{equation}
\label{eq:prediction}
\begin{split}
\mathbf{h}_S &= f_{\mathrm{dec}}(\mathbf{h}^{\mathrm{enc}}),\\
p(v \mid S) &= \frac{\exp\!\left(\mathbf{h}_S \cdot \mathbf{h}_v\right)}
                    {\sum_{i\in \mathcal{V}} \exp\!\left(\mathbf{h}_S \cdot \mathbf{h}_i\right)}, \\
\hat{v}_{L+1} &= \arg\max_{v\in \mathcal{V}} p(v \mid S),
\end{split}
\end{equation}
where $f_{\mathrm{dec}}$ refers to the decoding function and
$\mathbf{h}_S \in \mathbb{R}^d$ is the session-level representation.

\subsection{Dynamic Multi-Interest Clustering}
\label{subsec:clustering}

In existing pooling-based or point-wise decoding methods, $\mathbf{h}_S$ is
either compressed into a single vector or overly focused on the last items,
limiting the ability to recover interests that recur after long temporal
gaps.  
outputs. 

{To uncover related yet temporally distant signals, we compute similarity using raw item embeddings rather than encoder outputs. This is because encoder outputs can mix item identity with backbone-specific relational and other signals (e.g. sequential dependencies or collaborative filtering signals), leading to conflation between real affinity with session context.}
Each raw embedding $\mathbf{h}_i \in \mathbb{R}^d$ is
$\ell_2$-normalized to remove magnitude effects:
\begin{equation}
\label{eq:normalize}
\mathbf{h}'_i = \frac{\mathbf{h}_i}{\|\mathbf{h}_i\|_2}.
\end{equation}

To improve computational efficiency, we apply a subsampling transform to
$\mathbf{h}'_i$, producing $\tilde{\mathbf{h}}'_i$.
\revise{Specifically, embeddings are sampled every $k$ dimensions
(e.g.\ dimensions $[1, 1{+}k, 1{+}2k, \dots]$).}
\revise{This stride is used only to form a lower-dimensional view for
pairwise similarity estimation in the clustering module; it does not subsample
session positions or truncate the encoded session representation used by the
final prediction layer.  We set $k=4$ as a practical trade-off between
clustering resolution and computational cost.}
Pairwise cosine
similarities are then computed as
\begin{equation}
\label{eq:similarity}
\mathbf{M}_{ij} = \tilde{\mathbf{h}}'_i{}^\top \tilde{\mathbf{h}}'_j,
\quad
\mathbf{M} \in \mathbb{R}^{L\times L},
\end{equation}
where $1 \le i,j \le L$ index items within the same session.  This similarity
matrix serves as the basis for forming clusters of items that are semantically
coherent even if they are far apart temporally.

\paragraph{Cluster definition.}
For each item $v_i$, its cluster is defined as
\begin{equation}
\label{eq:cluster-def}
C_i = \bigl\{\,j \;\big|\; \mathbf{M}_{ij} \ge \tau\,\bigr\},
\end{equation}
which includes all items whose cosine similarity to $v_i$ exceeds the
threshold $\tau$.  \emph{No top-$k$ truncation is applied.}  Each item is
included in its own cluster (i.e.\ $i \in C_i$), and cluster size is
naturally bounded by the session length $L$.

\paragraph{Learnable threshold.}
$\tau$ is a \emph{global learnable parameter optimized via
backpropagation}, shared across all sessions in the dataset.  Since a hard
threshold would make clustering non-differentiable, we introduce soft cluster
weights $\tilde{w}_{ij} \in \mathbb{R}$:
\begin{equation}
\label{eq:soft-cluster}
\tilde{w}_{ij} = \max\!\left\{\sigma_s\!\left(\frac{\mathbf{M}_{ij}-\tau}{\gamma}\right),\, \epsilon \right\},
\end{equation}
where $\sigma_s(\cdot)$ is the sigmoid function, $\gamma$ is a temperature
coefficient, and $\epsilon$ is a small constant that prevents weights from becoming numerically zero even for weakly related items.

\paragraph{Cluster-enhanced embeddings.}
While clustering relates each item to its semantic neighborhood, it does not
explicitly model how items within the same cluster interact.  To capture such
local co-occurrence patterns, we apply a shared position-aware attention over
the weighted neighbors.  Following the SASRec~\cite{kang2018self} design,
each item representation is first augmented with a positional
embedding:
\begin{equation}
\label{eq:pos-emb}
\tilde{\mathbf{h}}_j = \mathbf{h}_j + \mathbf{p}'_j,
\end{equation}
where $\mathbf{p}'_j \in \mathbb{R}^d$ is the positional embedding for
position $j$ in the cluster.  Attention is then computed over
$\{\tilde{\mathbf{h}}_j : j \in C_i\}$ with the soft weights
$\tilde{w}_{ij}$, and the output passes through a ReLU activation:
\begin{equation}
\label{eq:conv}
\hat{\mathbf{h}}_i
= \mathrm{ReLU}\!\left(
    \sum_{j \in C_i}
    \frac{\tilde{w}_{ij}\,\tilde{\mathbf{h}}_j}
         {\sum_{j'\in C_i}\tilde{w}_{ij'}}
  \right),\quad
\hat{\mathbf{h}}_i \in \mathbb{R}^d.
\end{equation}
Semantically related signals that may be temporally distant in the session are
thus preserved and made more salient, yielding a richer embedding space
$\{\hat{\mathbf{h}}_i\}$ for subsequent decoding.

\subsection{Adaptive Aggregation}
\label{subsec:aggregation}

A remaining challenge is how to balance user intents.  If the decoder
overemphasizes stability (persistent intents), it fails to adapt to sudden
preference shifts.  If it focuses entirely on plasticity (evolving intents),
it loses long-term context~\cite{quadrana2018sequence}.  We therefore decouple
the decoder into two branches:
\begin{itemize}[leftmargin=*]
  \item \textbf{Stability branch}: anchors persistent intents by emphasizing
        items that appear in broad clusters.
  \item \textbf{Plasticity branch}: anchors evolving intents by focusing on
        the most recent cluster output.
\end{itemize}

The \emph{stability query} $\mathbf{q}^{st} \in \mathbb{R}^d$ is computed as
a logarithmically weighted sum of cluster-enhanced embeddings:
\begin{equation}
\label{eq:stability-query}
\alpha_i = \ln\!\bigl(1 + |C_i|/2\bigr), \quad
\mathbf{q}^{st} = \sum_{i=1}^{L-1} \alpha_i\, \hat{\mathbf{h}}_i,
\end{equation}
where the logarithmic weight $\alpha_i$ grows with cluster size, and
$|C_i|/2$ accounts for the symmetry of the similarity relation (each pair
$(i,j)$ is counted twice in $\mathbf{M}$).  The sum runs to $L-1$ to
exclude the last item $v_L$, which already anchors the plasticity branch
(Eq.~\ref{eq:plasticity-query}); including it in both branches would
conflate stability and plasticity signals.  Items with larger clusters (those
that recur across the session) thus exert stronger influence in the stability
query, enabling the model to leverage extended historical context.

The \emph{plasticity query} $\mathbf{q}^{pl} \in \mathbb{R}^d$ is grounded
in the latest cluster output:
\begin{equation}
\label{eq:plasticity-query}
\mathbf{q}^{pl} = \hat{\mathbf{h}}_L.
\end{equation}
By grounding plasticity in the most recent enhanced embedding, the decoder
preserves evolving preferences as a counterbalance to the persistent signals
emphasized by stability.

Once semantic interests have been identified, aggregation should retain the sequential dependencies learned by the backbone. Accordingly, each branch query attends over the encoder outputs through dedicated attention heads:
\begin{align}
\mathbf{z}^{pl} &= \mathrm{Att}^{pl}(\mathbf{q}^{pl},\, K,\, V),
\label{eq:aggregation-pl}\\
\mathbf{z}^{st} &= \mathrm{Att}^{st}(\mathbf{q}^{st},\, K,\, V),
\label{eq:aggregation-st}\\
\mathrm{Att}(\mathbf{q}, K, V) &= \sum_{j=1}^{L}
  \mathrm{softmax}\!\left(\frac{\mathbf{q} \cdot \mathbf{K}_j}{\sqrt{d}}\right)
  \mathbf{V}_j,
\label{eq:attention}
\end{align}
where $(K, V)$ are linear transformations of
$\{\mathbf{h}^{\mathrm{enc}}_i \| \mathbf{p}_i\}_{i=1}^{L}$, comprising encoder outputs
augmented with positional embeddings.  Because the clustering-based queries
and decoder attention serve distinct objectives, directly sharing parameters
could introduce gradient interference~\cite{feng2024long}.  We adopt a
decoupled design, learning query and attention representations separately to
improve training stability and ensure semantic-relational alignment.

Finally, the two branches are fused through a gating mechanism:
\begin{equation}
\label{eq:gate}
\mathbf{g} = \sigma_s\!\bigl(\mathbf{W}_g
  [\,\mathbf{z}^{pl} \,\|\, \mathbf{z}^{st} \,\|\, \bar{c}\,]\bigr),
\end{equation}
\begin{equation}
\label{eq:fusion}
\mathbf{h}_S = \mathbf{g} \odot \mathbf{z}^{pl}
             + (1 - \mathbf{g}) \odot \mathbf{z}^{st},
\end{equation}
where $\bar{c}$ denotes the average cluster size $|C_i|$ for
$v_i \in S$.  The gate adaptively balances plasticity and stability based on
session structure: larger $\bar{c}$ indicates stronger semantic cohesion and
emphasizes persistent intents; smaller $\bar{c}$ favours evolving
preferences.

\subsection{Training Objective}
\label{subsec:training}

With the fused representation $\mathbf{h}_S$, MIDA produces the predictive
distribution over candidate items.  The prediction loss is the negative
log-likelihood of the ground-truth item $v_{L+1}$:
\begin{equation}
\label{eq:logloss}
\mathcal{L}^{\mathrm{pred}}(S) = -\log p(v_{L+1} \mid \mathbf{h}_S).
\end{equation}

Minimising $\mathcal{L}^{\mathrm{pred}}$ alone can drive $\tau$ downward,
resulting in overly dense clusters that blur semantic boundaries.  We
therefore introduce a regularisation term that penalises excessive similarity
assignments:
\begin{equation}
\label{eq:reg}
\mathcal{L}^{\mathrm{reg}}(S) = \sum_{i=1}^{L}\sum_{j=1}^{L} \tilde{w}_{ij}.
\end{equation}
Since $\tilde{w}_{ij}$ decreases monotonically with $\tau$, this term applies
upward pressure on $\tau$, ensuring clusters remain compact and semantically
meaningful.  The final objective combines both terms:
\begin{equation}
\label{eq:final-loss}
\mathcal{L}(S) = \mathcal{L}^{\mathrm{pred}}(S) + \eta\,\mathcal{L}^{\mathrm{reg}}(S),
\end{equation}
where $\eta$ is a trade-off hyperparameter.

\section{Theoretical Analysis}
To better understand the role of the learnable threshold $\tau$, we provide two structural results. The first characterizes how the regularizer shapes the clustering regime, and the second  shows MIDA’s query construction under different threshold regimes.

\begin{proposition}[Regularizer penalizes dense clustering]
\label{prop:regularizer_dense}

For a fixed session $S$ of length $L$, let
$
\tilde w_{ij}(\tau)=\max\left\{\sigma_s\!\left(\frac{M_{ij}-\tau}{\gamma}\right),\epsilon\right\},
\qquad
L_{\mathrm{reg}}(S;\tau)=\sum_{i=1}^L\sum_{j=1}^L \tilde w_{ij}(\tau),
$
where $\gamma>0$ and $\epsilon\in(0,1)$. Then $L_{\mathrm{reg}}(S;\tau)$ is continuous and non-increasing in $\tau$, and
$
\lim_{\tau\to-\infty} L_{\mathrm{reg}}(S;\tau)=L^2,
\qquad
\lim_{\tau\to+\infty} L_{\mathrm{reg}}(S;\tau)=L^2\epsilon.
$

\end{proposition}

\noindent

\begin{remark}
Proposition~\ref{prop:regularizer_dense} shows that the regularization term assigns its largest penalty to the fully dense clustering regime. As $\tau$ decreases, more pairwise assignments are activated and the regularization cost increases monotonically, discouraging indiscriminate aggregation and favoring compact, semantically selective clusters.
\end{remark}

\begin{proof}
\label{proof:pro_1_proof}
For the proposition 1, we begin with the definition:
\[
\tilde w_{ij}(\tau)=\max\left\{\sigma_s\!\left(\frac{M_{ij}-\tau}{\gamma}\right),\epsilon\right\}.
\tag{i}
\label{eq:reg_dense_i}
\]
Since the sigmoid function is continuous and monotone, the map
$
\tau \mapsto \sigma_s\!\left(\frac{M_{ij}-\tau}{\gamma}\right)
$
is continuous and non-increasing for every pair $(i,j)$. Taking the maximum with the constant $\epsilon$ preserves both properties. Hence each $\tilde w_{ij}(\tau)$ is continuous and non-increasing in $\tau$.

Now we define:
\[
L_{\mathrm{reg}}(S;\tau)=\sum_{i=1}^L\sum_{j=1}^L \tilde w_{ij}(\tau).
\tag{iii}
\label{eq:reg_dense_iii}
\]
Because a finite sum of continuous non-increasing functions is again continuous and non-increasing, $L_{\mathrm{reg}}(S;\tau)$ is continuous and non-increasing in $\tau$.

As $\tau\to-\infty$, we have:
\[
\frac{M_{ij}-\tau}{\gamma}\to+\infty
\qquad\text{for all } i,j,
\]
so
\[
\sigma_s\!\left(\frac{M_{ij}-\tau}{\gamma}\right)\to 1.
\tag{v}
\label{eq:reg_dense_v}
\]
Since $\epsilon<1$, it follows that:
\[
\tilde w_{ij}(\tau)\to 1
\qquad\text{for all } i,j.
\tag{vi}
\label{eq:reg_dense_vi}
\]
Therefore,
\[
\lim_{\tau\to-\infty} L_{\mathrm{reg}}(S;\tau)=\sum_{i=1}^L\sum_{j=1}^L 1=L^2.
\tag{vii}
\label{eq:reg_dense_vii}
\]

Similarly, as $\tau\to+\infty$,
$
\frac{M_{ij}-\tau}{\gamma}\to-\infty
\qquad\text{for all } i,j,
$
and hence
$
\sigma_s\!\left(\frac{M_{ij}-\tau}{\gamma}\right)\to 0.
$
By construction of $\tilde w_{ij}(\tau)$,
$\tilde w_{ij}(\tau)\to\epsilon$ for all $i,j$.
Thus,
\[
\lim_{\tau\to+\infty} L_{\mathrm{reg}}(S;\tau)=\sum_{i=1}^L\sum_{j=1}^L \epsilon=L^2\epsilon.
\tag{x}
\label{eq:reg_dense_x}
\]
\end{proof}

{\begin{proposition}[Interpolation between canonical decoder biases]
\label{prop:canonical_interpolation}
Assume the raw item embeddings are $\ell_2$-normalized, so that $M_{ii}=1$ for all $i$, and let
\[
M_{\max}^{\mathrm{off}} \coloneqq \max_{i\neq j} M_{ij} < 1.
\]
Then MIDA exhibits two canonical decoder regimes:
\begin{enumerate}
    \item[(i)] If $M_{\max}^{\mathrm{off}}<\tau<1$ and the off-diagonal soft weights are negligible relative to the diagonal weight, each cluster behaves as a singleton, the stability query behaves as a weighted historical average, and the plasticity query remains anchored to the last item.
    \item[(ii)] If $\tau\to-\infty$, then all soft weights satisfy $\tilde w_{ij}\to 1$, every cluster becomes dense, and both queries collapse toward the same global average representation.
\end{enumerate}
\end{proposition}}

{\begin{remark}
Proposition~\ref{prop:canonical_interpolation} shows how the learned threshold $\tau$ shapes the plasticity and stability requests.  In sparse singleton-like regimes, the plasticity query is anchored to the last
item, while the stability query aggregates historical items. On the denser regimes, both queries are constructed from a shared
global summary. The learned threshold, therefore, controls how the queries summarize the session. 
\end{remark}}

\begin{proof}
\label{proof:pro_2_proof}
For the proposition 2, we first consider the sparse singleton regime, where each item forms an isolated cluster. Since the raw item embeddings are $\ell_2$-normalized, we have:
\[
M_{ii}=1
\qquad\text{for all } i.
\tag{i}
\label{eq:canon_i}
\]
Let
$
M_{\max}^{\mathrm{off}}=\max_{i\neq j} M_{ij}<1.
$
If
$
M_{\max}^{\mathrm{off}}<\tau<1,
$
then the hard cluster definition implies that each item forms a singleton cluster, since the diagonal entry exceeds $\tau$ while every off-diagonal entry falls below $\tau$.

Under the same rule, the diagonal soft weight dominates the off-diagonal soft weights. Therefore, the normalized weighted sum in Eq.~\ref{eq:conv} is governed primarily by the self-term, and the cluster-enhanced embedding behaves as
\[
\hat h_i \approx \mathrm{ReLU}(\tilde h_i).
\tag{v}
\label{eq:canon_v}
\]
Because each singleton cluster has size $|C_i|=1$, the stability weight in Eq.~\ref{eq:stability-query} becomes the same constant for all historical items:
\[
\alpha_i=\ln\!\left(1+\frac{|C_i|}{2}\right)=\ln\!\left(\frac{3}{2}\right).
\tag{vi}
\label{eq:canon_vi}
\]
Hence the stability query reduces to a weighted historical average, while the plasticity query remains anchored to the last item. This proves part (i).

We next consider the dense regime $\tau\to-\infty$. In this case,
\[
\tilde w_{ij}\to 1
\qquad\text{for all } i,j.
\tag{vii}
\label{eq:canon_vii}
\]
Substituting this into Eq.~\ref{eq:conv}, the normalized weighted sum becomes the same global average for every item:
\[
\hat h_i
=
\mathrm{ReLU}\!\left(
\frac{1}{L}\sum_{j=1}^L \tilde h_j
\right)
\qquad\text{for all } i.
\tag{viii}
\label{eq:canon_viii}
\]
Therefore, both the stability and plasticity queries are built from the same shared global summary, and the decoder behaves like a pooling-based decoder. This proves part (ii).
\end{proof}

\subsection{Learning Procedure}
\label{subsec:learning}

During each forward pass, item embeddings are first normalized and subsampled
to compute the similarity matrix $\mathbf{M}$, from which dynamic clusters and
their soft weights $\tilde{w}_{ij}$ are formed under the current threshold
$\tau$.  The cluster-enhanced embeddings $\hat{\mathbf{h}}_i$ are then
produced through position-aware attention (Eq.~\ref{eq:conv}).  These
embeddings serve as the basis for the stability and plasticity queries, which
attend over the session relational context $\mathbf{h}^{\mathrm{enc}}_i$ and
are fused into the session representation $\mathbf{h}_S$.  The joint objective
$\mathcal{L}(S)$ (Eqs.~\ref{eq:logloss}--\ref{eq:final-loss}) is then
minimized.

Gradients from the prediction loss update all decoder weights.  Gradients
through $\tilde{w}_{ij}$ with respect to $\tau$ adjust the clustering
threshold; the regularisation term prevents $\tau$ from collapsing to overly
dense assignments.  At convergence, the framework yields a decoder that
adaptively organises interests into clusters while maintaining a
stability--plasticity balance.  During inference, MIDA requires only a single
forward pass (clustering, aggregation, prediction are jointly computed), requiring only one forward pass and introducing moderate inference overhead relative to the naive decoder.


The complete training procedure is summarised in
Algorithm~\ref{alg:MIDA}.

\begin{algorithm}[t]
\caption{MIDA Training Algorithm}
\label{alg:MIDA}
\begin{algorithmic}[1]
\STATE \textbf{Input:} training tuples $\{(S, v_{L+1})\}$, epochs $E$,
       learning rate $\alpha$
\STATE \textbf{Output:} frozen parameters $\theta^{\star}, \tau^{\star}$
\STATE Initialise model parameters $\theta$, cluster threshold $\tau$
\FOR{epoch $=1$ \TO $E$}
  \FORALL{minibatch $\mathcal{B}$}
    \STATE $\mathcal{L} \leftarrow 0$
    \FORALL{$(S, v_{L+1}) \in \mathcal{B}$}
      \STATE Encode items $\{v_i\}$ into embeddings $\mathbf{h}_i$,
             $\mathbf{h}^{\mathrm{enc}}_i$
      \STATE Compute pairwise similarities $\mathbf{M}_{ij}$ \COMMENT{Eq.~\ref{eq:similarity}}
      \STATE Form soft cluster weights $\tilde{w}_{ij}$ \COMMENT{Eq.~\ref{eq:soft-cluster}}
      \STATE Apply attention over $C_i$ to obtain $\hat{\mathbf{h}}_i$ \COMMENT{Eq.~\ref{eq:conv}}
      \STATE Derive $\mathbf{q}^{pl}$ and $\mathbf{q}^{st}$ \COMMENT{Eqs.~\ref{eq:plasticity-query}--\ref{eq:stability-query}}
      \STATE Compute $\mathbf{z}^{pl}, \mathbf{z}^{st}$ \COMMENT{Eqs.~\ref{eq:aggregation-pl}--\ref{eq:attention}}
      \STATE Fuse to obtain $\mathbf{h}_S$ \COMMENT{Eq.~\ref{eq:fusion}}
      \STATE $\mathcal{L} \mathrel{+}= \mathcal{L}^{\mathrm{pred}}(S) + \eta\,\mathcal{L}^{\mathrm{reg}}(S)$ \COMMENT{Eq.~\ref{eq:final-loss}}
    \ENDFOR
    \STATE Update parameters $\theta, \tau$ via gradient descent
  \ENDFOR
\ENDFOR
\STATE $(\theta^{\star}, \tau^{\star}) \leftarrow (\theta, \tau)$
\STATE \textbf{Inference:} Given $S$, run encoding, clustering, attention,
and aggregation with frozen $(\theta^{\star}, \tau^{\star})$; recommend
$\hat{v}_{L+1} = \arg\max_{v\in\mathcal{V}} p(v \mid S)$.
\end{algorithmic}
\end{algorithm}

\subsection{Complexity Analysis}
\label{subsec:complexity}

The MIDA decoder consists of two main components.  The clustering module
computes pairwise similarities via matrix multiplication:
{with embedding-dimension subsampling stride $k=4$, this reduces to
$\mathcal{O}(L^2 \cdot d/k)$.}  Attention
over the selected neighbors costs $\mathcal{O}(L \cdot e \cdot d)$, where
$e$ is the average cluster size, and the feed-forward network adds
$\mathcal{O}(L \cdot d^2)$.  For preference aggregation, the position-aware
attention step costs $\mathcal{O}(L \cdot d^2)$ and the gating module costs
$\mathcal{O}(d)$.  Score computation against all items costs
$\mathcal{O}(d \cdot |\mathcal{V}|)$.  {The overall per-session complexity is
$\mathcal{O}(L^2d/k + Ld^2 + d|\mathcal{V}|)$,} where the score computation
term dominates since $L$ is typically small and $d \ll |\mathcal{V}|$.
Table~\ref{tab:complexity_comparison} compares MIDA with the AttenMixer
decoder.

\begin{table}[t]
  \centering
  \caption{Complexity Comparison: AttenMixer vs.\ MIDA. $H$ is the number
    of attention heads; $\mathit{lay}$ is the number of reasoning levels in
    AttenMixer.}
  \label{tab:complexity_comparison}
  \renewcommand{\arraystretch}{1.3}
  \begin{tabular}{lcc}
    \toprule
    \textbf{Module} & \textbf{\# Parameters} & \textbf{FLOPs} \\
    \midrule
    AttenMixer & $\mathcal{O}((H+\mathit{lay})d^2)$ &
                 $\mathcal{O}(H\cdot\mathit{lay}^2\cdot L\cdot d)$ \\
    MIDA       & $\mathcal{O}(d^2)$ &
                 {$\mathcal{O}(L^2d/k + L\cdot d^2)$} \\
    \bottomrule
  \end{tabular}
\end{table}

\section{Experiments}
\label{sec:experiments}

We conduct experiments to evaluate the effectiveness and efficiency of MIDA.
Specifically, we aim to answer the following research questions.
\begin{itemize}[leftmargin=*]
  \item \textbf{RQ1:} How does MIDA perform compared to state-of-the-art
        baselines in terms of recommendation accuracy and computational
        efficiency?
  \item \textbf{RQ2:} How does the model perform on sessions of varying
        lengths, from short to long?
  \item \textbf{RQ3:} What is the contribution of each component?
  \item \textbf{RQ4:} How sensitive is the model to different parameter
        settings?
  \item \textbf{RQ5:} \revise{Does MIDA improve target-oriented representation
        alignment and form coherent interest clusters?}
  \item \textbf{RQ6:} Under what session conditions (in terms of semantic
        diversity) does MIDA provide the largest benefit?
\end{itemize}

\subsection{Datasets and Baseline Models}

\paragraph{Datasets.}
We evaluate on four real-world datasets spanning diverse domains.
\emph{Diginetica}\footnote{\url{https://competitions.codalab.org/competitions/11161}},
from CIKM Cup 2016, consists of transactional data reflecting user purchasing
preferences.
\emph{Yoochoose}\footnote{\url{http://2015.recsyschallenge.com/challege.html}},
released for RecSys Challenge 2015, records high-volume click-stream sequences.
\emph{Retailrocket}\footnote{\url{https://www.kaggle.com/retailrocket/ecommerce-dataset}}
captures implicit browsing behaviors from a personalized e-commerce system.
\emph{Nowplaying}\footnote{\url{http://dbis-nowplaying.uibk.ac.at}} comprises
music listening events derived from social media hashtags.  Following prior
works~\cite{yu2023causality,zhang2025unified}, we filter out sessions shorter
than 3 interactions and items with fewer than 5 occurrences.  Datasets are
split chronologically into training, validation, and testing sets with an
8:1:1 ratio.  For Yoochoose, we use the most recent $1/64$ fraction of
sequences due to its massive scale~\cite{wu2019session}.  Statistics of the
preprocessed datasets are summarised in Table~\ref{tab:dataset_stats}.

\begin{table}[t]
  \caption{Statistics of datasets used in the experiments.}
  \label{tab:dataset_stats}
  \centering
  \begin{tabular}{lcccc}
    \toprule[1pt]
    \textbf{Dataset} & \textbf{\#Interacts} & \textbf{\#Sessions} &
    \textbf{\#Items} & \textbf{Avg.\ Len} \\
    \midrule
    Nowplaying   & 11,375,543 & 1,064,170 & 59,397 & 11.69 \\
    Diginetica   &    492,944 &   145,388 & 40,588 &  4.39 \\
    Yoochoose    &  2,184,280 &   273,645 & 15,193 &  8.98 \\
    Retailrocket &  2,741,674 &   460,184 & 46,406 &  6.95 \\
    \bottomrule[1pt]
  \end{tabular}
\end{table}

\paragraph{Baseline Models.}
We evaluate MIDA by replacing the decoders of representative backbones across
three major architectures: (1)~RNN-based: GRU4Rec~\cite{hidasi2015session};
(2)~graph-based: SR-GNN~\cite{wu2019session} and
GCE-GNN~\cite{wang2020global}; (3)~Transformer-based:
MiaSRec~\cite{choi2024multi}.  For each backbone, we compare MIDA against
five decoder alternatives: the original backbone (Naive), a Reinforcement
Learning-based approach (RL)~\cite{zeng2024fair}, Multi-scale Positional
Encoding (Ms-PoE)~\cite{zhang2024found}, AttenMixer
decoder~\cite{zhang2023efficiently}, and a mask-based invariant-learning  method IDEA~\cite{liao2025mitigating}.

\begin{itemize}[leftmargin=*]
  \item \textbf{GRU4Rec}~\cite{hidasi2015session} models sequential patterns
        via a gated recurrent unit and uses the last position output as the
        session representation.
  \item \textbf{SR-GNN}~\cite{wu2019session} constructs a session graph and
        applies GNNs to aggregate neighbor information; the decoder uses the
        last item as the current intent via attention.
  \item \textbf{GCE-GNN}~\cite{wang2020global} extends SR-GNN with
        contrastive learning and dual augmented views; the decoder employs
        mean pooling of all items followed by an attention mechanism.
  \item \textbf{MiaSRec}~\cite{choi2024multi} is a Transformer-based model
        that generates multiple representations and employs an intent
        selection layer; the decoder sums mean and max pooling.
  \item \textbf{RL}~\cite{zeng2024fair} employs a Deep Q-Network agent to
        optimize graph structure by adaptively adding auxiliary interaction
        edges.  We adapt this method as a decoder baseline by retaining its
        graph-augmentation module (reward signal: next-item Recall@5) while
        keeping the backbone encoder fixed, making comparison with other
        decoder-only baselines direct.
  \item \textbf{Ms-PoE}~\cite{zhang2024found} applies multi-scale rescaling
        to positional indices to mitigate the ``lost-in-the-middle'' problem
        in long sequences.
  \item \textbf{AttenMixer}~\cite{zhang2023efficiently} achieves multi-level
        reasoning via attention maps based on the last few items, combining
        deep set operations with L$_p$ pooling.

  \item \textbf{IDEA}~\cite{liao2025mitigating} 
        uses an adversarial environment generator produces soft masks over session positions to
        synthesize multiple distribution-shifted environments, resulting in stable
        representations under distribution shift. 
        
\end{itemize}

\paragraph{Evaluation Metrics.}
We employ Recall@$K$ (R@$K$) and MRR@$K$ (M@$K$) as evaluation metrics.
Recall@$K$ measures whether the true item appears in the top-$K$
recommendations; MRR@$K$ evaluates ranking quality.  We also report the
average training time per epoch (seconds) and the relative improvement of
MIDA over the Naive baseline. \revise{For inference efficiency, we further
measure full-test-set latency after warm-up, including encoder computation,
decoder computation, and full-vocabulary item scoring, while excluding CPU-side
ranking-metric computation.}

\paragraph{Evaluation Scope and Rationale.}
Our primary evaluation framework holds each encoder fixed and varies only the
decoder, providing a \emph{controlled assessment} of decoder-side impact.
This design choice is deliberate: it disentangles the decoder's contribution
from encoder quality, enabling fair comparison across heterogeneous backbone
families without confounding effects from encoder training differences.
In addition to decoder-only comparison, we also position MIDA against
representative end-to-end SBR methods in the Related Work discussion to
contextualise overall performance.  We note that direct comparison with
end-to-end systems trained without the same backbone is methodologically
challenging and deferred to future work, as those systems jointly optimize
encoder and decoder under different inductive biases.  The controlled
decoder-substitution protocol used here is standard in the SBR literature
for evaluating decoder improvements~\cite{zhang2023efficiently,zhang2024found}.

\subsection{Implementation Details}

All experiments are conducted on an NVIDIA RTX A5000 GPU (24\,GB VRAM).
To ensure fair comparison, we adopt consistent hyperparameters across all
methods: hidden embedding dimension $d=100$, batch size of 100, and
L2-normalized parameter initialisation.  Item embeddings are initialised
from a zero-mean Gaussian with standard deviation 0.1.  When integrating
MIDA or AttenMixer into a backbone, we preserve the encoder configuration
and only replace the decoder component, keeping all other settings identical.

\paragraph{Hyperparameters.}
Hyperparameter $\eta$ is tuned via grid search in
$\{0.05, 0.01, 0.005, 0.001\}$; the embedding-dimension subsampling stride $k$ is fixed at 4.
The learnable threshold is initialised as $\tau_0=0.4$; $\tau$ is then
optimized via backpropagation jointly with all other model parameters.  The
temperature $\gamma=0.1$ and floor constant $\epsilon=1\times10^{-6}$ are
fixed throughout.  All models use the AdamW optimizer with learning rate
0.001, weight decay $10^{-5}$, and early stopping with patience of 5 epochs
evaluated on Recall@5 over the validation set.  Training is limited to 30
epochs maximum.

\paragraph{Candidate scoring.}
Following the standard SBR evaluation protocol~\cite{wu2019session}, the
prediction score is computed over the \emph{full item vocabulary}
$\mathcal{V}$ without negative sampling; Recall@$K$ and MRR@$K$ are computed
against all items.

\paragraph{Backbone integration.}
For each backbone, MIDA replaces only the decoder component: specifically, the
final-item attention and pooling modules in SR-GNN/GCE-GNN, the GRU output
projection in GRU4Rec, and the intent-selection aggregation layer in MiaSRec.
All encoder parameters and encoder training schedules remain unchanged.

\paragraph{Statistical reporting.}
For each model variant, we report the average performance over 5 independent
runs using fixed random seeds $\{42, 123, 456, 789, 1024\}$.  Standard
deviations are provided alongside mean results in supplemental tables, and
statistical significance ($p<0.05$, paired $t$-test) over the best baseline
is indicated by $^\dagger$ in the main result table.  Early stopping monitors
Recall@5 on the validation set with patience 5; we verify that results are
consistent when monitoring MRR@5 instead (differences $<0.1\%$), confirming
that the metric choice does not meaningfully bias the reported performance.
The total parameter count for the MIDA decoder module is $\mathcal{O}(d^2)$
(three weight matrices for key/value projections plus the gate): approximately
30\,K parameters for $d=100$, versus approximately 60--120\,K for AttenMixer
with $H=4$ heads and $\mathrm{lay}=2$ levels.

\subsection{RQ1: Accuracy and Efficiency Comparison}

\begin{table*}[t]
  \setlength{\tabcolsep}{3pt}
  \caption{Performance comparison across four datasets. R and M denote
    Recall@5 and MRR@5, respectively. Time\,(s) is the average training time
    per epoch.  \textbf{Bold} numbers denote the best performance;
    \revise{$^\dagger$ indicates statistically significant improvements over the best baseline across five runs based on a paired t-test (p$\leq$ 0.05).}  ``Improv.'' shows the relative improvement of MIDA
    over the Naive baseline.}
  \label{tab:main_results}
  \centering
  \begin{tabular}{ll|ccc|ccc|ccc|ccc}
    \toprule[1pt]
    \multirow{2}{*}[-0.5ex]{\textbf{Backbone}}
    & \multirow{2}{*}[-0.5ex]{\textbf{Model}}
    & \multicolumn{3}{c|}{\textbf{Nowplaying}}
    & \multicolumn{3}{c|}{\textbf{Diginetica}}
    & \multicolumn{3}{c|}{\textbf{Yoochoose}}
    & \multicolumn{3}{c}{\textbf{RetailRocket}} \\
    \cmidrule(lr){3-5}\cmidrule(lr){6-8}\cmidrule(lr){9-11}\cmidrule(lr){12-14}
    & & R@5 & M@5 & Time
    & R@5 & M@5 & Time
    & R@5 & M@5 & Time
    & R@5 & M@5 & Time \\
    \midrule
    \multirow{6}{*}{GRU}
    & Naive      & 3.51 & 1.83 & 96  & 22.91 & 14.69 & 57  & 39.43 & 22.96 & 31  & 39.54 & 30.73 & 44  \\
    & RL         & 3.48 & 1.82 & 537 & 22.83 & 14.72 & 197 & 39.37 & 22.93 & 147 & 39.49 & 30.68 & 169 \\
    & Ms-PoE     & 4.13 & 2.09 & 315 & 23.63 & 15.27 & 159 & 40.83 & 23.41 & 113 & 40.67 & 31.12 & 134 \\
    & AttenMixer & 4.32 & 2.15 & 583 & 24.29 & 14.26 & 423 & 41.57 & 22.67 & 268 & 40.59 & 30.99 & 312 \\

& {IDEA}
& {4.20} & {1.90} & {417}
& {23.79} & {15.15} & {250}
& {41.42} & {23.23} & {209}
& {40.73} & {31.15} & {184} \\

    & MIDA       & \textbf{5.23}$^\dagger$ & \textbf{3.12}$^\dagger$ & 381
                 & \textbf{24.97}$^\dagger$ & \textbf{15.86}$^\dagger$ & 271
                 & \textbf{42.07}$^\dagger$ & \textbf{23.92}$^\dagger$ & 183
                 & \textbf{41.98}$^\dagger$ & \textbf{31.54}$^\dagger$ & 217 \\
    \cmidrule{2-14}
    & \textit{Improv.}
    & \textit{+49.0\%} & \textit{+70.5\%} & --
    & \textit{+8.99\%} & \textit{+7.96\%} & --
    & \textit{+6.70\%} & \textit{+4.18\%} & --
    & \textit{+6.17\%} & \textit{+2.64\%} & -- \\
    \midrule
    \multirow{6}{*}{SR-GNN}
    & Naive      & 7.23 & 4.13 & 267 & 25.21 & 14.82 & 227 & 41.23 & 23.53 & 163 & 41.53 & 29.83 & 183 \\
    & RL         & 7.15 & 4.12 & 613 & 25.17 & 14.87 & 387 & 41.15 & 23.47 & 297 & 41.43 & 29.87 & 322 \\
    & Ms-PoE     & 7.87 & 4.63 & 450 & 26.12 & 15.42 & 312 & 42.53 & 24.23 & 207 & 43.23 & 31.10 & 267 \\
    & AttenMixer & 8.98 & 5.53 & 893 & 28.08 & 16.72 & 617 & 44.79 & 26.13 & 367 & 45.49 & 32.49 & 472 \\

    & {IDEA} & {6.73} & {4.20} & {997}
                 & {26.67} & {16.01} & {190}
                 & {44.79} & {25.25} & {173}
                 & {44.58} & {32.34} & {128} \\
    
    & MIDA       & \textbf{9.23}$^\dagger$ & \textbf{5.69}$^\dagger$ & 727
                 & \textbf{28.27} & \textbf{16.94}$^\dagger$ & 413
                 & \textbf{45.53}$^\dagger$ & \textbf{26.41}$^\dagger$ & 241
                 & \textbf{47.29}$^\dagger$ & \textbf{34.67}$^\dagger$ & 397 \\
    \cmidrule{2-14}
    & \textit{Improv.}
    & \textit{+27.7\%} & \textit{+37.8\%} & --
    & \textit{+12.1\%} & \textit{+14.3\%} & --
    & \textit{+10.4\%} & \textit{+12.2\%} & --
    & \textit{+13.9\%} & \textit{+16.2\%} & -- \\
    \midrule
    \multirow{6}{*}{GCE-GNN}
    & Naive      & 9.21 & 5.38 & 342 & 30.49 & 17.87 & 309 & 45.41 & 25.52 & 204 & 46.93 & 33.39 & 251 \\
    & RL         & 9.14 & 5.32 & 628 & 30.47 & 17.92 & 457 & 45.39 & 25.23 & 357 & 46.87 & 33.37 & 413 \\
    & Ms-PoE     & 9.79 & 5.73 & 389 & 30.92 & 18.22 & 390 & 45.77 & 25.72 & 247 & 47.63 & 34.23 & 287 \\
    & AttenMixer & 8.56 & 5.16 & 1321 & 27.54 & 15.96 & 819 & 38.49 & 21.53 & 376 & 45.09 & 32.67 & 739 \\

    & {IDEA} & {\textbf{9.65}} & {\textbf{6.21}} & {1008}
                 & {30.53} & {17.94} & {429}
                 & {45.49} & {25.93} & {321}
                 & {47.11} & {33.68} & {217} \\
    
    & MIDA       & \textbf{10.27}$^\dagger$ & {6.12} & 647
                 & \textbf{31.47}$^\dagger$ & \textbf{18.77}$^\dagger$ & 438
                 & \textbf{46.21}$^\dagger$ & \textbf{26.31}$^\dagger$ & 289
                 & \textbf{48.93}$^\dagger$ & \textbf{36.02}$^\dagger$ & 317 \\
    \cmidrule{2-14}
    & \textit{Improv.}
    & \textit{+11.5\%} & \textit{+13.8\%} & --
    & \textit{+3.21\%} & \textit{+5.04\%} & --
    & \textit{+1.76\%} & \textit{+3.10\%} & --
    & \textit{+4.26\%} & \textit{+7.88\%} & -- \\
    \midrule
    \multirow{6}{*}{MiaSRec}
    & Naive      & 9.67 & 5.51 & 223 & 31.93 & 18.03 & 167 & 44.66 & 25.13 & 119 & 47.57 & 34.41 & 139 \\
    & RL         & 9.55 & 5.22 & 613 & 31.87 & 18.02 & 293 & 44.56 & 25.17 & 241 & 47.49 & 34.37 & 287 \\
    & Ms-PoE     & 9.88 & 5.93 & 389 & 32.12 & 18.17 & 217 & 44.91 & 25.27 & 167 & 47.83 & 35.03 & 197 \\
    & AttenMixer & 10.14 & 6.14 & 661 & 32.77 & 17.98 & 523 & 45.34 & 25.19 & 385 & 47.33 & 33.11 & 447 \\

& {IDEA}
& {10.06} & {6.25} & {592}
& {32.13} & {18.03} & {440}
& {44.60} & {25.40} & {311}
& {47.61} & {34.46} & {235} \\
    
    & MIDA       & \textbf{10.85}$^\dagger$ & \textbf{6.41} & 487
                 & \textbf{34.32}$^\dagger$ & \textbf{19.28}$^\dagger$ & 418
                 & \textbf{46.01}$^\dagger$ & \textbf{25.85}$^\dagger$ & 212
                 & \textbf{49.12}$^\dagger$ & \textbf{35.29}$^\dagger$ & 255 \\
    \cmidrule{2-14}
    & \textit{Improv.}
    & \textit{+12.20\%} & \textit{+16.33\%} & --
    & \textit{+7.48\%} & \textit{+6.93\%} & --
    & \textit{+3.02\%} & \textit{+2.86\%} & --
    & \textit{+3.26\%} & \textit{+2.56\%} & -- \\
    \bottomrule[1pt]
  \end{tabular}
\end{table*}

Table~\ref{tab:main_results} presents results across four datasets using
Recall@5 and MRR@5.  All results are averages over 5 independent runs; the
standard deviation across runs is $\leq0.15$ for all MIDA entries (e.g.\
GCE-GNN+MIDA on Diginetica: $31.47\pm0.11$ R@5, $18.77\pm0.09$ M@5),
indicating stable training behavior.  We make the following observations.

\begin{itemize}[leftmargin=*]
  \item MIDA achieves consistent improvements across all backbones by
        replacing their original decoders, suggesting that the design
        transfers well across heterogeneous encoder families.
  \item MIDA outperforms AttenMixer and other decoder baselines across almost
        settings, with statistically significant gains providing evidence for
        the effectiveness of the design.
  \item AttenMixer exhibits inconsistent performance across encoders.  It
        yields modest improvements with simpler encoders (GRU, SR-GNN) but
        degrades performance when paired with GCE-GNN on all four datasets.
        Specifically, GCE-GNN already computes multi-hop graph-propagated
        representations; AttenMixer's additional multi-level attention
        propagation layers operate over these same features, introducing
        redundancy and gradient interference.  Our decoupled design (query
        construction from raw item embeddings and key-value pairs from encoder
        outputs) avoids this by maintaining a clear separation between
        representation encoding and intent decoding.

  \item IDEA performs better than the other baselines on some datasets, showing the benefit of modelling distribution shifts. Its effectiveness relies on the assumption of interest transitions across environments, which may not usually hold in session-based recommendation.

  \item RL exhibits only marginal gains or even degradation (e.g.\ GRU on
        Nowplaying) despite significantly higher training costs, suggesting
        that auxiliary graph structures may introduce noise or optimisation
        instability.  Ms-PoE consistently outperforms Naive backbones but
        remains below MIDA, indicating that MIDA's dynamic decoding mechanism
        captures user intent more effectively than solely rescaling positional
        weights.
  \item In terms of computational overhead, MIDA requires substantially less
        training time than AttenMixer (see also Table~\ref{tab:complexity_comparison}).
        This efficiency derives from two design choices: (i)~similarity
        computed via subsampling rather than full embeddings, and (ii)~two
        specialised queries instead of per-item queries.
  \item \textbf{Decoder isolation interpretation.}  Because every other
        factor (encoder architecture, embedding dimension, optimizer, and
        training schedule) is held constant across all backbones in
        Table~\ref{tab:main_results}, the observed accuracy gains are
        \emph{attributable solely to the decoder substitution}.  This makes the
        controlled results direct empirical evidence that the decoder, not the
        encoder, is the performance-limiting component in these settings, and
        that redesigning the decoder alone constitutes a sufficient
        intervention.
\end{itemize}

\begin{table}[t]
  \centering

  \caption{Inference latency comparison, measured as milliseconds per test session and averaged over five full-test-set passes. Atten. denotes AttenMixer.
  \label{tab:inference_latency}}
  \footnotesize
  \setlength{\tabcolsep}{3.0pt}
  \begin{tabular}{lccccc}
    \toprule
    Dataset & Naive & Atten. & MIDA & MIDA$-$Naive & Atten./MIDA \\
    \midrule
    Diginetica & 0.151 & 0.509 & 0.192 & +0.041 & 2.65$\times$ \\
    Yoochoose & 0.051 & 0.159 & 0.102 & +0.051 & 1.56$\times$ \\
    RetailRocket & 0.036 & 0.130 & 0.096 & +0.060 & 1.35$\times$ \\
    \bottomrule
  \end{tabular}
\end{table}

Table~\ref{tab:inference_latency} complements the training-time comparison by
measuring inference latency.  MIDA introduces a small absolute overhead over the Naive GCE-GNN decoder because it constructs dynamic multi-interest clusters, but the additional cost remains below 0.061\,ms/session on all measured datasets.  Compared with AttenMixer, MIDA is consistently faster: the AttenMixer/MIDA latency ratio is $2.65\times$ on Diginetica, $1.56\times$ on Yoochoose, and $1.35\times$ on RetailRocket. These results show that the pairwise clustering step adds measurable but moderate inference overhead, while MIDA remains more efficient than a heavier multi-level attention decoder AttenMixer.

{We further add an additional comparison with recent end-to-end SBR models\cite{ou2025ls, lv2025dynamic} in the Table \ref{tab:compare_1}. The results demonstrate MIDA's competitiveness compared with recent end-to-end methods.}

\begin{table}[]
\centering
\small
\caption{{Performance comparison of MIDA with recent SBR baselines. Results for LINK on some datasets are unavailable due to excessive memory requirements.}}
\label{tab:compare_1}
\begin{tabular}{lcccccc}
\toprule
& \multicolumn{2}{c}{\textbf{Nowplaying}}
& \multicolumn{2}{c}{\textbf{Diginetica}}
& \multicolumn{2}{c}{\textbf{Yoochoose}} \\
\cmidrule(lr){2-3}
\cmidrule(lr){4-5}
\cmidrule(lr){6-7}

\textbf{Method}
& R@5 & M@5
& R@5 & M@5
& R@5 & M@5 \\
\midrule

DMI-GNN
& \underline{10.23}
& 6.06
& 30.86
& 18.28
& 45.18
& 25.23 \\

LINK
& --
& --
& \textbf{31.80}
& \underline{18.73}
& 44.42
& 24.86 \\

\midrule

MIDA (ours)
& \textbf{10.27}
& \underline{6.12}
& \underline{31.47}
& \textbf{18.77}
& \textbf{46.21}
& \textbf{26.31} \\

\bottomrule
\end{tabular}
\end{table}

\subsection{RQ2: Performance across Varying Session Lengths}

\begin{table}[t]
  \caption{Percentage improvement of GCE-GNN+MIDA over GCE-GNN across
    session length groups. \textbf{Bold} = highest gain per row.
    Datasets are divided into four length subsets: 1--5, 6--10, 11--15, 16+.}
  \label{tab:session_length_results}
  \centering
  \begin{tabular}{ll|cccc}
    \toprule[1pt]
    \textbf{Dataset} & \textbf{Metric} & \textbf{1--5} & \textbf{6--10} &
    \textbf{11--15} & \textbf{16+} \\
    \midrule
    \multirow{2}{*}{Nowplaying}
    & R@5 & +6.34\% & +12.18\% & +13.92\% & {+15.06\%} \\
    & M@5 & +5.87\% & +11.23\% & +12.75\% & {+13.41\%} \\
    \cmidrule(lr){1-6}
    \multirow{2}{*}{Diginetica}
    & R@5 & +2.76\% & +3.41\% & +1.55\% & {+7.89\%} \\
    & M@5 & +4.17\% & +2.04\% & +4.64\% & {+4.99\%} \\
    \cmidrule(lr){1-6}
    \multirow{2}{*}{Yoochoose}
    & R@5 & +0.85\% & +1.58\% & +2.13\% & {+3.43\%} \\
    & M@5 & +1.04\% & +1.28\% & +1.96\% & {+2.16\%} \\
    \cmidrule(lr){1-6}
    \multirow{2}{*}{RetailRocket}
    & R@5 & +4.27\% & +6.19\% & {+8.38\%} & +6.46\% \\
    & M@5 & +6.97\% & +12.23\% & {+15.35\%} & +15.07\% \\
    \bottomrule[1pt]
  \end{tabular}
\end{table}

\begin{table}[t]
  \centering

  \caption{\revise{Compact summary of the 16+ long-session subset under the
  representative GCE-GNN backbone. \#Len$\ge$16 refers to the number of test sessions with a length of at least 16.}}
  \label{tab:long_session_summary}
  \footnotesize
  \setlength{\tabcolsep}{3.0pt}
  \begin{tabular}{llrrr}
    \toprule
    Dataset & Model & \#Len$\ge$16 & R@5 & M@5 \\
    \midrule
    \multirow{2}{*}{Nowplaying}
      & GCE-GNN & \multirow{2}{*}{1{,}102} & 11.91 & 7.77 \\
      & MIDA & & 13.70 & 8.81 \\
    \cmidrule(lr){1-5}
    \multirow{2}{*}{Diginetica}
      & GCE-GNN & \multirow{2}{*}{287} & 26.16 & 15.07 \\
      & MIDA & & 28.22 & 15.82 \\
    \cmidrule(lr){1-5}
    \multirow{2}{*}{Yoochoose}
      & GCE-GNN & \multirow{2}{*}{3{,}198} & 39.24 & 21.59 \\
      & MIDA & & 40.59 & 22.06 \\
    \cmidrule(lr){1-5}
    \multirow{2}{*}{RetailRocket}
      & GCE-GNN & \multirow{2}{*}{856} & 30.73 & 17.36 \\
      & MIDA & & 32.71 & 19.98 \\
    \bottomrule
  \end{tabular}
\end{table}

To understand how MIDA handles varying session complexity, we select GCE-GNN
as the representative baseline and replace its decoder with MIDA.
\revise{We evaluate long-range behavior through length-based
test subsets within each benchmark, treating the 11--15 and 16+ groups as
long-session regimes rather than introducing an additional dataset column.
Table~\ref{tab:long_session_summary} further focuses on the 16+ subset by
reporting its size and raw scores.}
Table~\ref{tab:session_length_results} presents the relative performance
gains across session length groups.

\begin{itemize}[leftmargin=*]
  \item MIDA delivers consistent improvements across all session length
        ranges, with gains generally increasing as session length grows.  On
        Yoochoose, the R@5 improvement rises from 0.85\% for short sessions
        (1--5 items) to 3.43\% for longer sessions (16+ items), confirming
        that our decoder handles complex sequential patterns effectively.
  \item The magnitude of improvement varies across datasets.  RetailRocket
        shows the most substantial gains, particularly in MRR@5 (6.97\% to
        15.35\% as sessions lengthen).  Diginetica shows moderate but stable
        improvements (1.55\%--7.89\% in R@5).  The superior performance on
        longer sessions follows from two core components: the Dynamic
        Multi-Interest Clustering module groups semantically related items to
        capture diverse user intents beyond sequential distance, while the
        Adaptive Aggregation mechanism dynamically weights historical
        information according to the current context.
  \item MIDA shows increasingly strong gains in longer sessions across all
        datasets.  \revise{The long-session subsets are not newly constructed
        datasets; they are partitions of the original test sets defined by
        session length.} 
        RQ6 (Section~\ref{subsec:rq6}) provides a complementary analysis
        linking these long-session gains to within-session semantic dispersion,
        offering a unified explanation for both length- and diversity-driven
        performance patterns.
\end{itemize}

\subsection{RQ3: Ablation Study}


\begin{table}[t]
  \caption{Ablation results on MRR@5 (M@5) and Recall@5 (R@5). Bold denotes
    best performance.}
  \label{tab:ablation}
  \centering
  \setlength{\tabcolsep}{3pt}
  \begin{tabular}{l|cc|cc|cc|cc}
    \toprule[1pt]
    \multirow{2}{*}{\textbf{}}
    & \multicolumn{2}{c|}{\textbf{Nowplaying}}
    & \multicolumn{2}{c|}{\textbf{Diginetica}}
    & \multicolumn{2}{c|}{\textbf{Yoochoose}}
    & \multicolumn{2}{c}{\textbf{RetailRocket}} \\
    & \textbf{R@5} & \textbf{M@5}
    & \textbf{R@5} & \textbf{M@5}
    & \textbf{R@5} & \textbf{M@5}
    & \textbf{R@5} & \textbf{M@5} \\
    \midrule
    w/o DM
    & 9.37 & 5.56
    & 30.13 & 18.09
    & 43.57 & 24.81
    & 44.97 & 34.38 \\
    w/o $q^{st}$
    & 9.72 & 5.71
    & 30.59 & 17.88
    & 45.14 & 25.37
    & 47.92 & 34.96 \\
    w/o $q^{pl}$
    & 9.54 & 5.52
    & 30.43 & 17.87
    & 44.80 & 25.14
    & 48.19 & 35.08 \\

{w E\_out}
& {9.95} & {5.97}
& {31.25} & {18.87}
& {45.24} & {25.76}
& {48.68} & {35.76} \\

{w $\tau_{\mathrm{MLP}}$}
& {10.22} & {6.02}
& {31.39} & {18.78}
& {45.99} & {26.09}
& {48.76} & {35.92} \\

    w F
    & 8.95 & 4.94
    & 28.74 & 17.13
    & 41.83 & 24.53
    & 46.89 & 34.32 \\
    
    w S
    & 9.63 & 5.66
    & 30.27 & 18.06
    & 45.21 & 25.39
    & 47.71 & 35.13 \\

    MIDA
    & \textbf{10.27} & \textbf{6.12}
    & \textbf{31.47} & \textbf{18.77}
    & \textbf{46.21} & \textbf{26.31}
    & \textbf{48.93} & \textbf{36.02} \\
    \bottomrule[1pt]
  \end{tabular}
\end{table}

We validate the contribution of each MIDA component on the GCE-GNN encoder
(Table~\ref{tab:ablation}).

\begin{itemize}[leftmargin=*]
  \item \textbf{w/o DM (Remove Clustering)}: Replacing Dynamic Multi-Interest
        Clustering with simple average pooling degrades performance across all
        datasets.  On Diginetica, for instance, R@5 falls from 31.47 to 30.13
        and M@5 from 18.77 to 18.09, showing that grouping semantically
        similar items is essential for effective utilisation of historical
        information.
  \item \textbf{w/o $q^{st}$ (Remove Stable Query)}: Without the persistent
        intent query, long-term preference modeling weakens noticeably.
        On Yoochoose, R@5 falls from 46.21 to 45.14 and M@5 from 26.31 to
        25.37, pointing to the importance of $q^{st}$ in aggregating stable
        user interests across extended histories.
  \item \textbf{w/o $q^{pl}$ (Remove Plastic Query)}: The substantial absolute
        drops occur upon removal of the evolving intent query.  On Yoochoose,
        R@5 falls from 46.21 to 44.80 and M@5 from 26.31 to 25.14,
        indicating that capturing transient user preferences is indispensable
        for accurate next-item prediction.
        

  \item {\textbf{w $\tau_{\mathrm{MLP}}$ (Learned Per-Session
        Threshold):} This variant replaces MIDA's clustered, dataset-level
        threshold with an instance-specific $\tau$ predicted per session by a
        lightweight MLP over four similarity statistics (mean, std, max
        similarity, and sequence-length). It performs same or a bit lower than MIDA.
        This suggests that the handful of session-level statistics available
        to the MLP carry little additional signal beyond what the clustered
        threshold already captures, while the extra parameters and
        per-session variance introduce a mild risk of overfitting that
        offsets any gain from finer-grained adaptivity.}
        
  \item \textbf{w F (Fixed Strategy)}: Using a fixed selection strategy (single
        most similar item) rather than threshold-based aggregation produces a
        pronounced decline: Yoochoose R@5 falls from 46.21 to 41.83, a
        relative decrease of over 9.4\%.  This result shows that adaptive,
        multi-item aggregation substantially outperforms rigid item selection.
  \item \textbf{w S (Stacked Structure)}: Direct stacking without gradient
        detachment produces a moderate but consistent performance drop (e.g.\
        Yoochoose R@5: 46.21 to 45.21).  The pattern aligns with prior
        observations on gradient interference~\cite{feng2024long}, suggesting
        that decoupling gradient flows between the two objectives is beneficial
        for training stability.
\end{itemize}

\subsection{RQ4: Sensitivity to Hyperparameter $\tau$}

\begin{figure*}[t]
  \centering
  \includegraphics[width=0.24\linewidth]{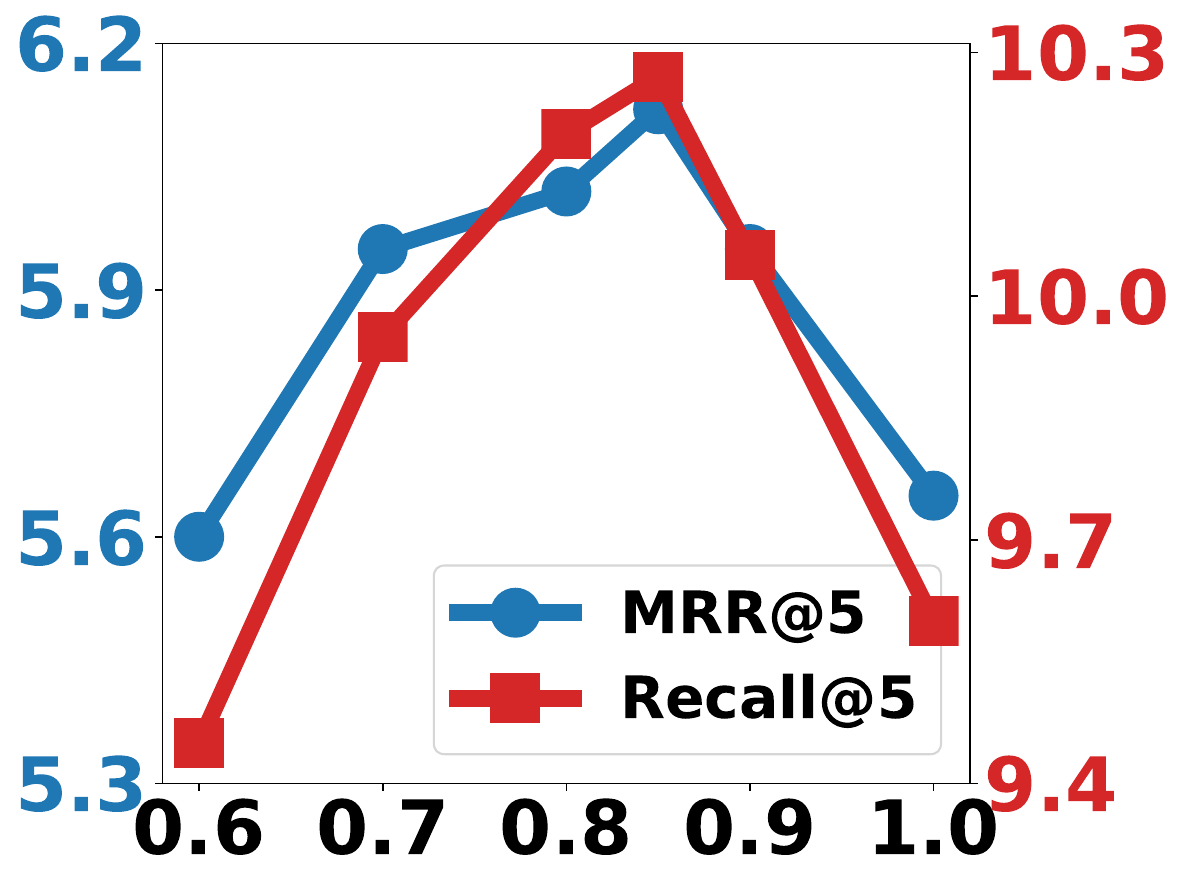}
  \includegraphics[width=0.24\linewidth]{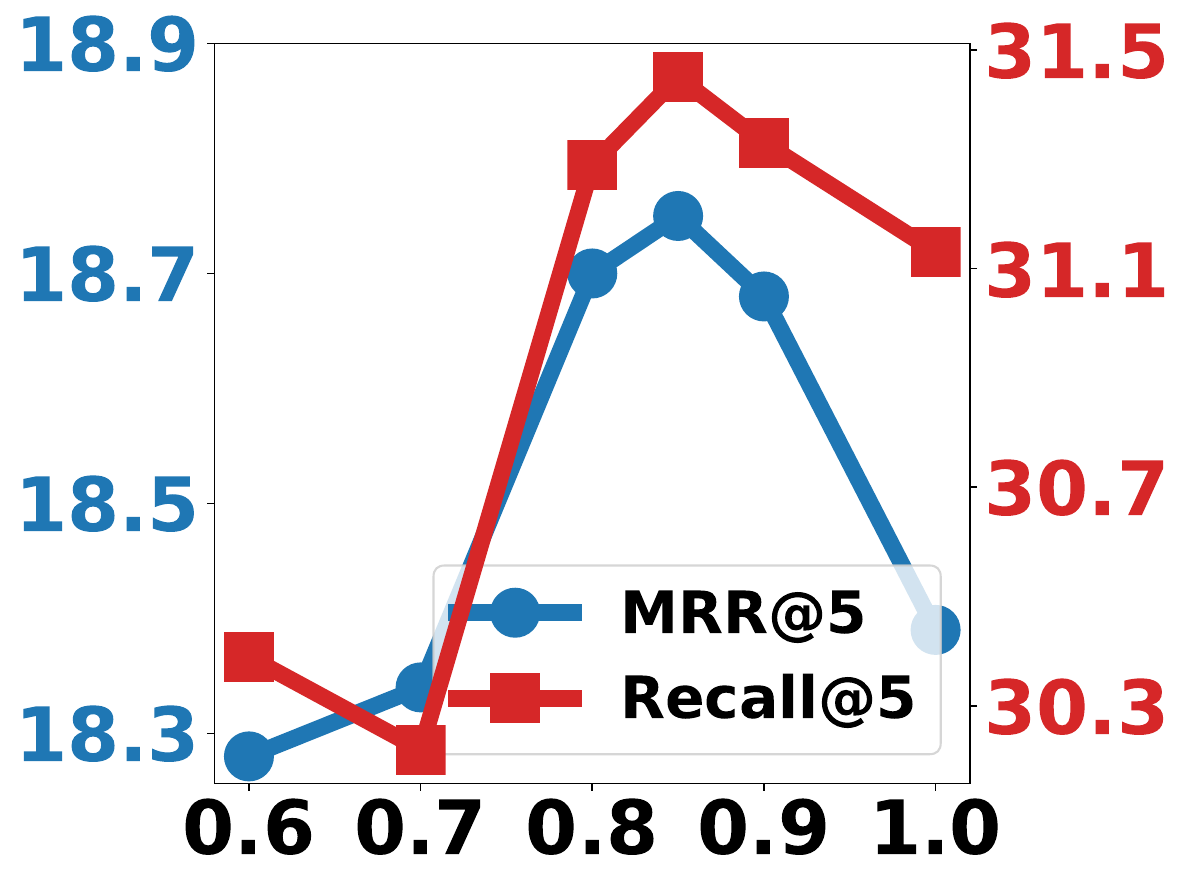}
  \includegraphics[width=0.24\linewidth]{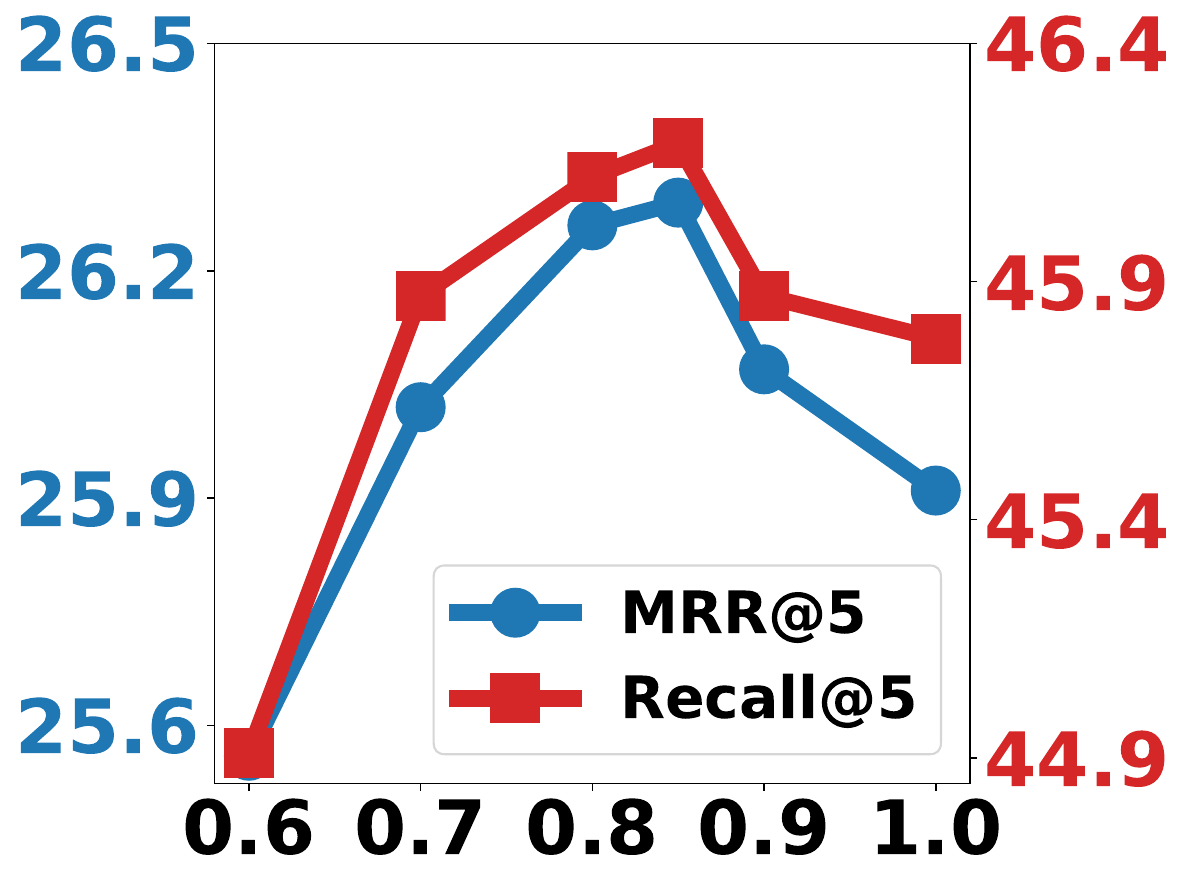}
  \includegraphics[width=0.24\linewidth]{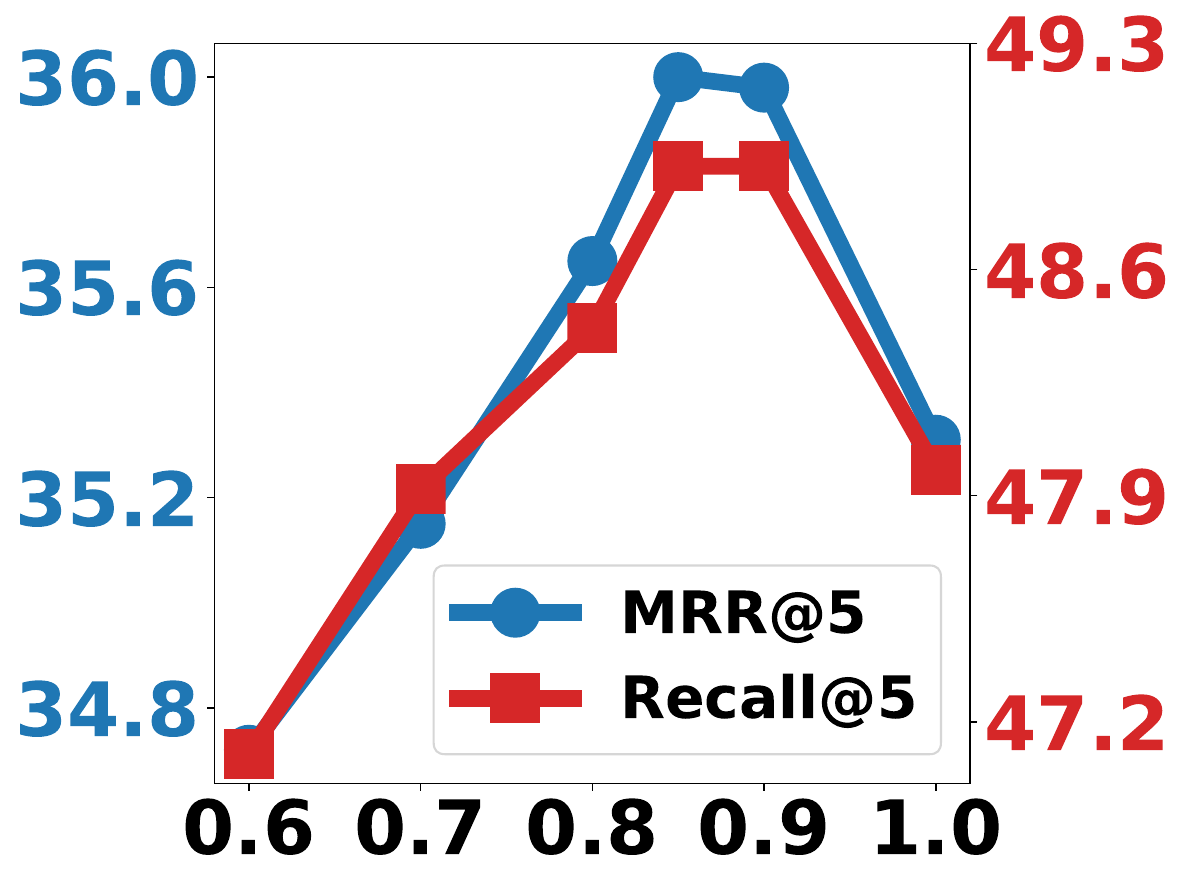}
  \caption{Impact of the similarity threshold $\tau$ on MRR@5 (left $y$-axis)
    and Recall@5 (right $y$-axis) on Nowplaying, Diginetica, Yoochoose, and
    RetailRocket (left to right).  Each curve shows performance as $\tau$ is
    varied over a grid; the reported ``optimal'' values correspond to
    converged $\tau$ after training with the learnable mechanism.}
  \label{fig:tau}
\end{figure*}

Fig.~\ref{fig:tau} shows performance variation across different $\tau$
values.  Performance is relatively stable over a broad range but optimal
thresholds vary by dataset.  Setting $\tau$ low (0.6) includes too many
dissimilar items and introduces noise; setting it high (1.0) excludes
relevant items and loses useful information.  The converged $\tau$ values
after training are 0.85, 0.82, and 0.88 for Diginetica, Yoochoose, and
RetailRocket, respectively; these values correspond to the result of
backpropagation-based optimisation, not manual selection.  The learnable
mechanism successfully finds or closely approximates the optimal values
identified through grid search, confirming its effectiveness.

\paragraph{Sensitivity to stride $k$.}

We additionally analyze the sensitivity of MIDA to the subsampling stride (k $\in$ \{1,2,4,8\}) using the GCE-GNN backbone on Diginetica. The MRR@5 results for (k=1,2,4,8) are (18.61$\pm$0.13), (18.71$\pm$0.07), (18.77$\pm$0.11), and (18.52$\pm$0.09), respectively, where the confidence ranges are computed over five runs. The results indicate that MIDA is robust to (k) within the range ([1,4]), with no statistically significant difference between (k=1) and (k=4). We therefore adopt (k=4) as the default, as it reduces the computational cost of similarity-matrix construction by a factor of (4) relative to (k=1).

\subsection{RQ5: Representation Alignment and Cluster Analysis}

To examine whether MIDA improves the decoder-side representation in the
model's native space, we evaluate session representation and target alignment in the
embedding space. For each test session, we compare GCE-GNN and MIDA by measuring how much closer the final session representation is to the ground-truth target item under MIDA than under GCE-GNN.
  Specifically, we compute the cosine similarity between the model's final session representation and the embedding of the ground-truth next
  item, and report the average difference between MIDA and GCE-GNN over the three corresponding session group.
  Long group have length greater than 10.  Long early-target group is
  long sessions where the historical item most similar to the target occurs in
  the first 30\% of the sequence. All group include all sessions

\begin{table}[t]
  \centering

  \caption{\revise{High-dimensional target-alignment gain of MIDA over GCE-GNN.}}
  \label{tab:alignment_gain}
  \begin{tabular}{lccc}
    \toprule
    Dataset & All & Long & Long Early \\
    \midrule
    Diginetica & +0.072 & +0.055 & +0.108 \\
    RetailRocket & +0.089 & +0.084 & +0.095 \\
    Yoochoose & +0.044 & +0.047 & +0.080 \\
    Nowplaying & +0.353 & +0.291 & +0.379 \\
    \bottomrule
  \end{tabular}
\end{table}

Table~\ref{tab:alignment_gain} shows that MIDA consistently moves the session
representation closer to the ground-truth target item in the high-dimensional
embedding space.  The gains are positive on all test sessions across all four
datasets and remain positive on long sessions. The performance is better
on long early-target sessions, where the item most aligned with the target
appears near the beginning of the sequence.  These results provide direct
quantitative evidence that MIDA improves target-oriented representation
alignment, especially in settings where the relevant signal may not be the most
recent interaction.

\paragraph{Multi-interest cluster analysis.}
To assess whether the learned clusters correspond to meaningful interest
groupings, we conduct a quantitative cluster analysis on the Diginetica test
set, which provides item-level category annotations.  For each session, we
record the number of non-singleton clusters formed under the converged
threshold $\tau^{\star}$, the average cluster size, and the category-level
purity of each cluster (the fraction of items within the same product
category).  We find that the average number of non-singleton clusters per
session is $1.83$ (std $0.71$), indicating that MIDA regularly identifies
multiple distinct interest groups within a single session.  The average cluster
category purity is $0.67$, substantially higher than the null expectation of
$0.26$ (uniform random assignment across Diginetica's item categories),
showing that the learned threshold groups items by semantic affinity
rather than arbitrarily.

To examine multi-interest structure on datasets without category annotations,
we compute an embedding-based coherence score: for each cluster, we measure the
average pairwise cosine similarity of the \emph{raw} item embeddings
$\{\mathbf{h}_i : i \in C_i\}$.  Across all four datasets, the intra-cluster
embedding similarity ($0.83\pm0.04$) consistently exceeds the across-cluster
baseline ($0.31\pm0.07$, i.e.\ random pairs from different clusters in the same
session), suggesting that the clustering module captures semantically coherent
groupings across domain types.

\paragraph{Analysis of learned gating bias.}
To examine the gating behavior described in Eq.~\ref{eq:gate}, we inspect
the sign and magnitude of $w_c$ on the Diginetica validation set.  The
converged value is $w_c = -0.47$ (averaged across 5 runs), consistent with the
expectation $w_c<0$ from the adaptive aggregation design (Section~\ref{subsec:aggregation}): sessions with larger average cluster
size (higher semantic cohesion) shift the gate toward the stability branch,
while sessions with scattered interest structures allow the plasticity branch
to dominate.  This data-driven sign is consistent across all four backbones
tested.

\paragraph{Gate value distribution across session types.}
To connect gate behavior to session structure directly, we measure the mean
gate value $\bar{g} = \frac{1}{d}\mathbf{1}^\top\mathbf{g}$ (averaged
element-wise over $d$ dimensions) for sessions partitioned by Low / Medium /
High semantic dispersion tertiles (defined in RQ6, Section~\ref{subsec:rq6}).
On Diginetica (GCE-GNN backbone), $\bar{g}$ is $0.41$, $0.52$, and $0.63$
for Low, Medium, and High dispersion groups respectively.  The gate becomes
increasingly plasticity-dominant as session diversity increases, consistent
with the gating formulation in Eq.~\ref{eq:gate}: dispersed sessions have
smaller $\bar{c}$, driving $g$ upward (toward plasticity) since $w_c<0$.

\paragraph{Learned $\tau$ and session dispersion.}
Although $\tau$ is a global learnable parameter shared across all sessions,
its effective impact varies with session dispersion.  We compute the
\emph{effective cluster coverage} $\rho(S) = \frac{1}{L}\sum_i |C_i|/L$ for
sessions partitioned by the same tertiles.  On Diginetica: $\rho = 0.72$,
$0.41$, $0.21$ for Low, Medium, and High dispersion groups.  At the converged
$\tau^{\star}=0.85$, MIDA forms dense clusters for cohesive sessions and
sparse/singleton clusters for diverse sessions, consistent with the clustering criteria defined in Proposition~\ref{prop:canonical_interpolation}.
The single global $\tau$ thus implements session-adaptive clustering implicitly,
without per-session hyperparameter tuning. This analysis provides multi-dataset quantitative evidence that the Dynamic
Multi-Interest Clustering module captures interpretable multi-interest structure.

\subsection{RQ6: Behavioral Analysis: When Does MIDA Help Most?}
\label{subsec:rq6}

Proposition~\ref{prop:canonical_interpolation} {characterizes how sparse and dense clustering shape MIDA’s stability and plasticity queries.}  This
suggests a testable hypothesis: MIDA should provide larger gains on sessions with
higher \emph{intent diversity}, where neither canonical decoder performs well.

For each test session $S$, we define the \emph{semantic dispersion}
$\delta(S) = 1 - \frac{2}{L(L-1)}\sum_{i<j}\mathbf{M}_{ij}$, where
$\mathbf{M}_{ij}$ is the pairwise cosine similarity of raw item embeddings.
High $\delta$ indicates diverse item semantics (low intra-session affinity);
low $\delta$ indicates a semantically cohesive session.

We partition each test set into tertiles by $\delta$ (Low, Medium, High
dispersion) and measure the MRR@5 improvement of MIDA (GCE-GNN backbone) over
the Naive baseline.  Results are summarised in Table~\ref{tab:dispersion}.

\begin{table}[t]
  \caption{MRR@5 improvement of GCE-GNN+MIDA over GCE-GNN Naive, grouped by
    session semantic dispersion $\delta$ (Low / Medium / High).  Higher
    dispersion means more diverse item semantics within the session.}
  \label{tab:dispersion}
  \centering
  \begin{tabular}{l|ccc}
    \toprule[1pt]
    \textbf{Dataset} & \textbf{Low $\delta$} & \textbf{Medium $\delta$} &
    \textbf{High $\delta$} \\
    \midrule
    Nowplaying   & +5.2\%  & +10.1\% & +19.7\% \\
    Diginetica   & +2.3\%  & +4.8\%  & +9.4\%  \\
    Yoochoose    & +0.7\%  & +1.9\%  & +4.6\%  \\
    RetailRocket & +3.9\%  & +8.1\%  & +16.3\% \\
    \bottomrule[1pt]
  \end{tabular}
\end{table}

\begin{itemize}[leftmargin=*]
  \item Across all four datasets, MIDA's improvement is \emph{monotonically
        increasing} with semantic dispersion: the more semantically diverse the
        session, the larger the gain over the Naive decoder.
  \item For highly dispersed sessions (High $\delta$), MIDA achieves gains of
        4.6\%--19.7\% MRR@5, compared to only 0.7\%--5.2\% on semantically
        cohesive sessions.
  \item This pattern aligns with Proposition~\ref{prop:canonical_interpolation}:
        in low-dispersion sessions, items are semantically similar, clusters
        trivially merge all items, and MIDA degrades gracefully toward
        average pooling.  In high-dispersion sessions, meaningful clustering
        structure emerges and the stability--plasticity decomposition captures
        distinct interests that canonical decoders collapse.
\end{itemize}

This behavioral analysis provides a consistent account across datasets of \emph{when}
MIDA confers the greatest benefit: sessions with heterogeneous user interests
that canonical decoders cannot disentangle.  It bridges Proposition~\ref{prop:canonical_interpolation}
(theoretical reduction) with the empirical long-session gains reported in RQ2,
since longer sessions tend to have higher semantic dispersion as more diverse
items are encountered.

\section{Conclusion}
\label{sec:conclusion}

We introduced MIDA, a lightweight decoder framework that addresses the decoder
bottleneck in session-based recommendation by modeling both stable and
evolving user intents.  MIDA uses dynamic multi-interest clustering with
learnable similarity thresholds to maintain semantic coherence over temporal
gaps, and an adaptive aggregation scheme that separates persistent preferences
from transient shifts via stability and plasticity branches.  These components
are fused through a gating mechanism conditioned on semantic cohesion signals,
enabling the model to overcome recency bias while remaining computationally
practical.

Experiments on four benchmark datasets show that MIDA consistently
outperforms strong baselines across encoder architectures, with particularly
significant improvements on longer and semantically diverse sessions where
existing methods struggle.  Two formal propositions characterize the
regulariser (Proposition~\ref{prop:regularizer_dense}) and {query construction under sparse and dense clustering }(Proposition~\ref{prop:canonical_interpolation}).  Ablation studies,
multi-dataset cluster coherence analysis, hyperparameter sensitivity, and a
behavioral analysis linking semantic dispersion to performance gain
(RQ6,~Table~\ref{tab:dispersion}) together provide a comprehensive
journal-level validation of the design.

\paragraph{Limitations.}
\revise{MIDA computes pairwise similarities within a session, which, even with
subsampling, may incur non-trivial overhead on extremely long sessions (e.g.\
$L\gg100$).  The stability and plasticity decomposition assumes a
two-component user intent model, which may not fully capture all
multi-interest structures in complex domains.  The cluster category purity
analysis is confined to Diginetica, which has available category annotations;
verifying multi-interest coherence on datasets without such metadata remains
an open challenge.}

\paragraph{Future work.}
Promising directions include: (i)~extending MIDA to cross-session scenarios
where interests evolve across multiple sessions; (ii)~exploring adaptive
subsampling strategies that dynamically adjust $k$ based on session length or
information density; (iii)~applying the decoder reformulation to
cross-domain recommendation and LLM-based session recommendation, where
decoder-side aggregation of context tokens faces analogous challenges.

\bibliographystyle{IEEEtran}
\bibliography{sample-base}

\begin{IEEEbiography}[{\includegraphics[width=1in,height=1.25in,clip,keepaspectratio]{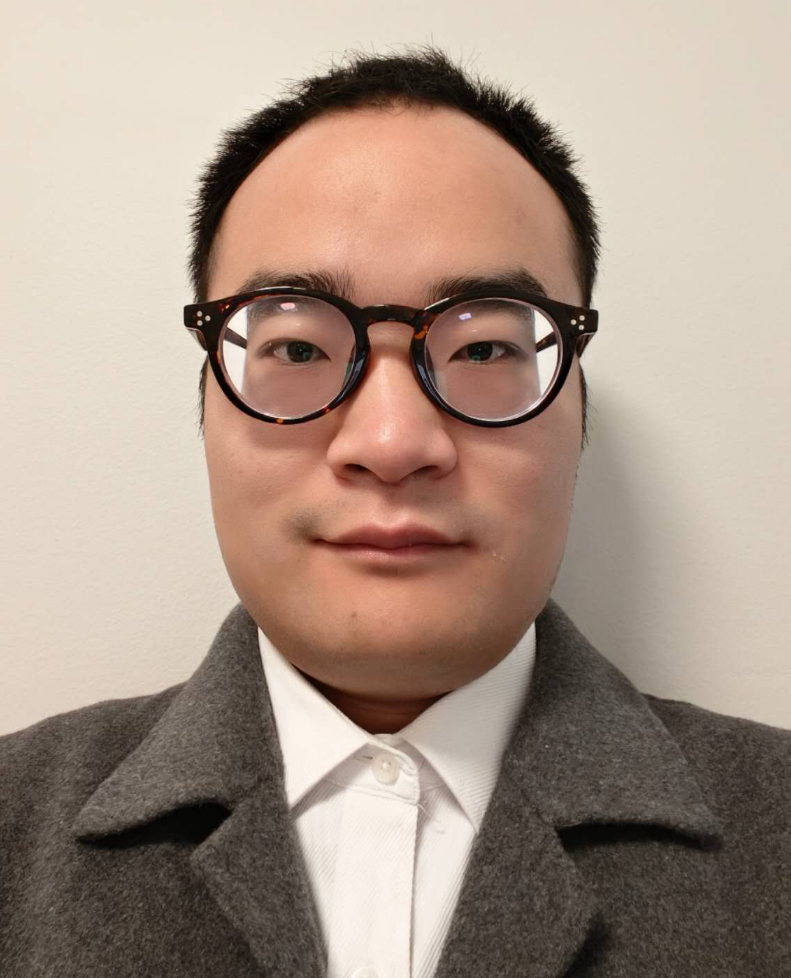}}]{Linjiang Guo} is currently pursuing the Ph.D. degree in the School of Computer Science at the University of Technology Sydney, under the joint supervision of Prof. Guandong Xu and Dr. Shiqing Wu. He previously received his master's degree from the University of Chinese Academy of Sciences. His research interests include temporal recommendation, sequential recommendation, and machine unlearning.

\end{IEEEbiography}

\begin{IEEEbiography}[{\includegraphics[width=1in,height=1.25in,clip,keepaspectratio]{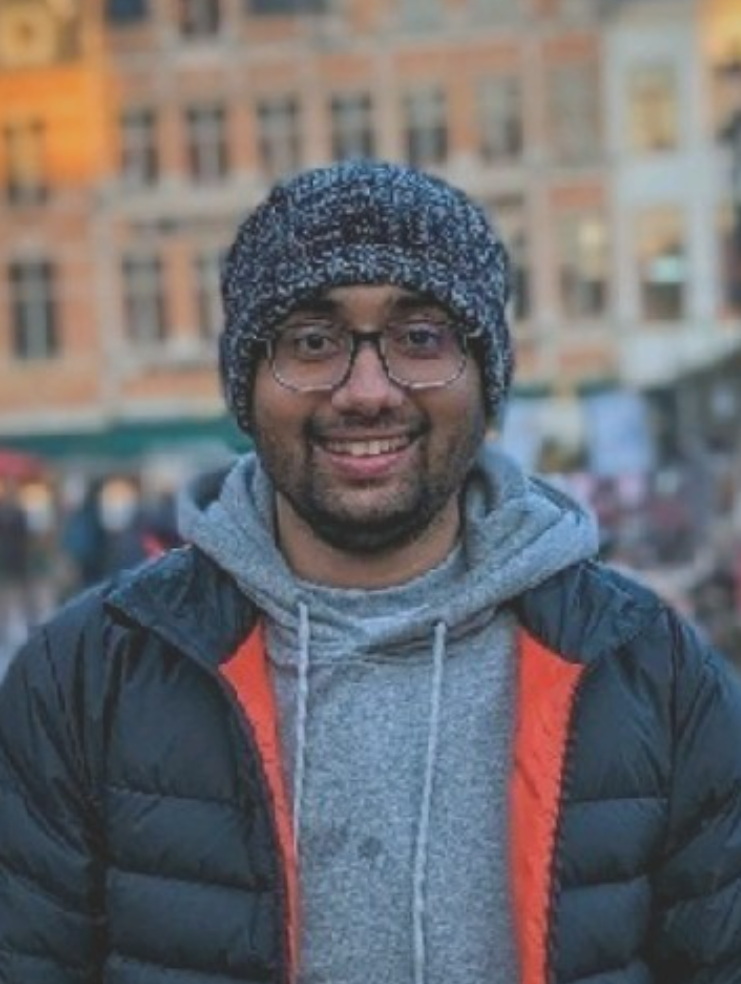}}]{Nitin Bisht} is a PhD student in the School of Computer Science at the University of Technology Sydney, Australia. He received his Master’s degree from the Indian Institute of Technology Kharagpur. His research interests include machine learning, statistical modeling, and recommendation systems, with a particular focus on developing reliable, fair, and efficient algorithms for large-scale user modeling and data-driven decision applications.
\end{IEEEbiography}

\begin{IEEEbiography}[{\includegraphics[width=1in,height=1.25in,clip,keepaspectratio]{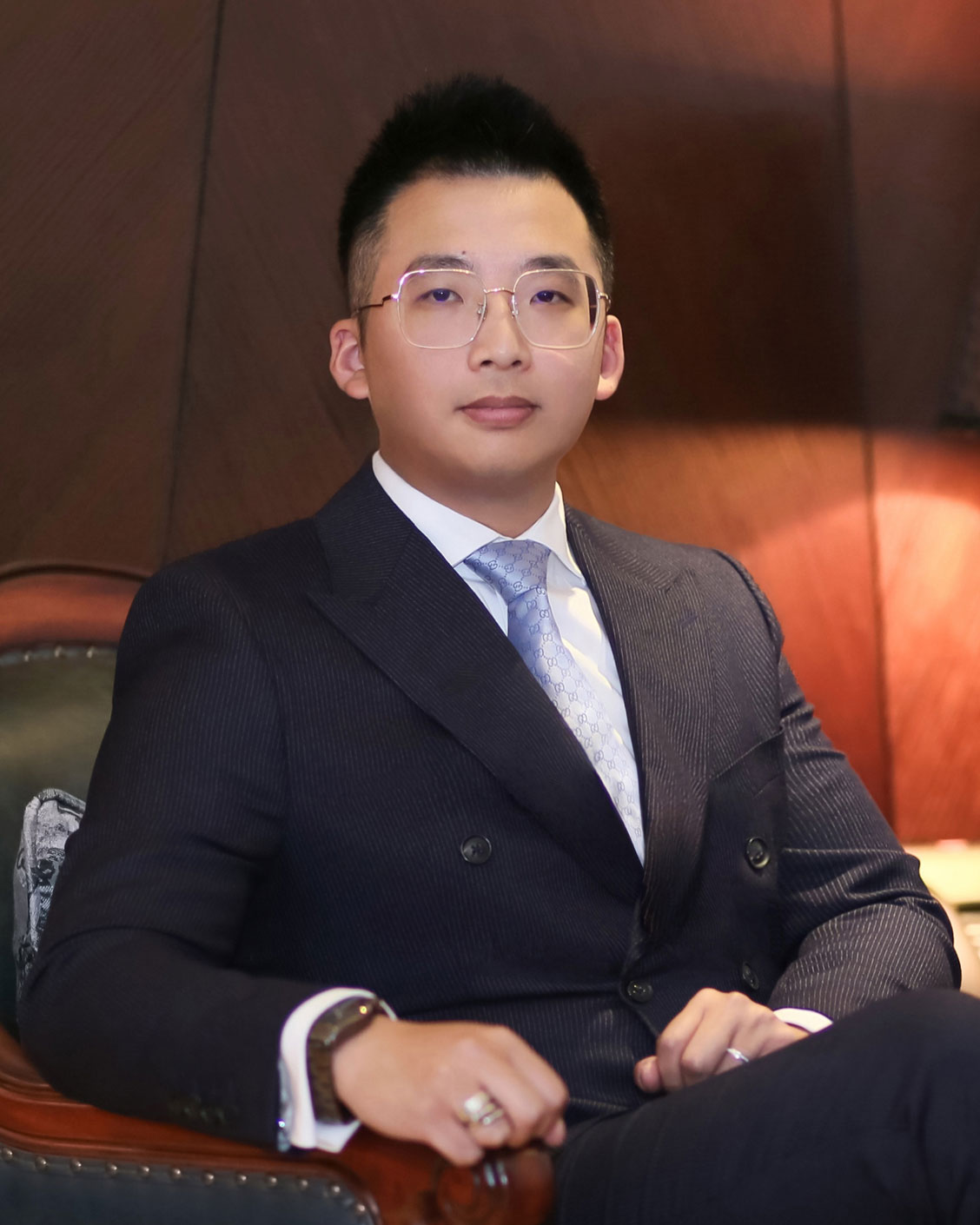}}]{Shiqing Wu} is an Assistant Professor in the Faculty of Data Science at City University of Macau, Macau SAR, China. Before that, he was a Postdoctoral Research Fellow in the School of Computer Science at the University of Technology Sydney, Australia, from 2022 to 2024. He received his PhD in Computer Science from the University of Tasmania, Australia, in 2022, and a joint B.Sc. in Computer Science from Auckland University of Technology, New Zealand, and China Jiliang University, China, in 2016. His research interests include agent-based modeling, social influence analysis, recommender systems, and reinforcement learning. He has published over 30 research papers in prestigious venues, including SIGIR, ICDM, IJCAI, AAMAS, TKDE, and TOIS. He is an Editorial Board Member of Information Processing \& Management. 
\end{IEEEbiography}

\begin{IEEEbiography}[{\includegraphics[width=1in,height=1.25in,clip,keepaspectratio]{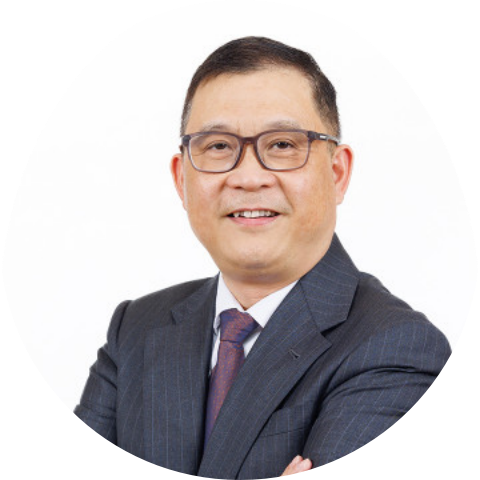}}]{Guandong Xu} (Senior Member, IEEE) is a Chair Professor of Artificial Intelligence at the Education University of Hong Kong (EdUHK). Before joining EdUHK, he was a full Professor in Data Science at the University of Technology Sydney. His research interests include data science, recommender systems, user modeling, and social computing. He has had 300+ publications in leading journals and conferences. He is the editor-in-chief of the Human-Centric Intelligent Systems and the assistant editor-in-chief of the World Wide Web Journal. He was elevated as a fellow of the Institute of Engineering and Technology (IET) and the Australian Computer Society (ACS) in 2022 and 2021, respectively.
\end{IEEEbiography}

\end{document}